\documentclass[
3p,times]{elsarticle}

\usepackage{amssymb}
\usepackage{amsthm}
\usepackage{setspace}
\usepackage{amsfonts}
\usepackage{amssymb}
\usepackage{amsmath}
\usepackage{booktabs}
\usepackage{caption}
\usepackage{subcaption}
\usepackage{dsfont}
\usepackage{enumerate}
\usepackage{times} 
\usepackage{todonotes}
\usepackage{graphicx}
\usepackage{latexsym}
\usepackage{pgf,pgfarrows,pgfnodes,pgfautomata,pgfheaps}
\usepackage{tikz}
\usetikzlibrary{arrows,automata,fit,matrix}

\newcommand{\pdefigs}[1]{}

\pdefigs{
\usepackage[crop=off]{auto-pst-pdf}
}

\newcommand{\dsl}{\displaystyle}
\newcommand{\lapl}{\bigtriangleup^d}

\newcommand{\lact}[1]{\xrightarrow{\tiny #1}}
\newcommand{\dt}{\Delta t}

\newcommand{\ds}{\Delta s}

\newtheorem{proposition}{Proposition}
\newtheorem{theorem}{Theorem}

\newtheorem{corollary}{Corollary}
\newtheorem{example}{Example}

\begin{document}

\begin{frontmatter}

\title{Spatial Fluid Limits for Stochastic Mobile Networks}

\author{Max Tschaikowski}
\ead{max.tschaikowski@imtlucca.it}

\author{Mirco Tribastone}
\ead{mirco.tribastone@imtlucca.it}
\cortext[cor1]{Corresponding author}

\address{IMT --- Institute for Advanced Studies Lucca, Italy}

\begin{abstract}
We consider Markov models of large-scale networks where nodes are characterized by their local behavior and by a mobility model over a two-dimensional lattice. By assuming random walk, we prove convergence to a system of partial differential equations (PDEs) whose size depends neither on the lattice size nor on the population of nodes. This provides a macroscopic view of the model which approximates discrete stochastic movements with continuous deterministic diffusions. We illustrate the practical applicability of this result by modeling a network of mobile nodes with on/off behavior performing file transfers with connectivity to 802.11 access points. By means of an empirical validation against discrete-event simulation we show high quality of the PDE approximation even for low populations and coarse lattices. In addition, we confirm the computational advantage in using the PDE limit over a traditional ordinary differential equation limit where the lattice is modelled discretely, yielding speed-ups of up to two orders of magnitude.
\end{abstract}


\end{frontmatter}


\section{Introduction}\label{sec:introduction}

This paper is concerned with the modeling and analysis of large-scale computer and communication networks with mobile nodes. We consider an expressive model of a node which is assumed to evolve over a set of \emph{local states}. These may represent, for instance, distinct phases of a node's behavior that distinguish between periods of activity and inactivity in a wireless network or, with a richer granularity, they may be related to different states in a communication protocol. We are interested in understanding how the interplay between the node's local behavior and its mobility pattern have an impact on the overall system's performance.

In this paper, we consider \emph{Markov population processes} as the starting point of our analysis. These represent a general class of continuous-time Markov chains (CTMCs), characterized by a state descriptor which is given by a vector where each element gives the population of nodes in a specific local state. In general, such stochastic models do not scale well with increasing system sizes because they are based on an explicit discrete-state representation. This makes the analysis very difficult computationally and poses a detriment to model-assisted parameter-space exploration and capacity planning.  However, under mild assumptions, a Markov population process may yield a \emph{fluid approximation} as a system of ordinary differential equations (ODEs). This can be shown to be the deterministic limit behavior that holds when the number of nodes goes to infinity; see, e.g.,~\cite{kurtz-1970}. Fluid limits have been developed for a wide range of models of distributed and networked systems, including, for instance, load balancing~\cite{Ganesh09,gast2010workstealing}, optical switches~\cite{DBLP:conf/sigmetrics/HoudtB12}, virtualized environments~\cite{Anselmi20111207}, and peer-to-peer networks~\cite{Zhou:2011:SOU:1993744.1993779}. In~\cite{benaim-leboudec} Bena\"im and Le Boudec offer a general framework for a Markov model of \emph{interacting agents} evolving through a set of \emph{local states}. In all cases, while the Markov process has a state space with cardinality that is exponential in the number of agents (in the worst case), the fluid limit only depends linearly on the number of local states. This procedure is very attractive in practice because it enables computationally efficient solutions, when the local state space is small, that are typically very accurate for large populations.


In the aforementioned models, however, spatial effects are not explicitly present because the system is inherently static, or because it is a convenient simplification to abstract away from them for the purposes of model tractability~\cite{Zhang20072867}.  Unfortunately, explicitly taking into account mobility and space in a fluid model generally leads to  a rapid growth of the local state space. In this paper, for instance, we deal with population processes defined by a CTMC with a random walk (RW) model over a two-dimensional topology.
We consider a typical partitioning of the spatial domain into \emph{regions} (or \emph{patches}, as usually called in theoretical biology~\cite{PatchReview}). Within each region an assumption of homogeneity is made, and spatial effects are incorporated as transitions across regions~\cite{DBLP:conf/sigmetrics/ChaintreauBR09}. If the dynamics of each region admits a fluid limit, then the overall system behavior is given by $L(K+1)^2$ ODEs, where, according to the notation of this paper, $L$ is the local state space size and $(K+1)^2$ is the total number of regions. Such dependence is simply due to the fact that the model must keep track of the total population of nodes in each of the $L$  local states, \emph{at each region}.

We consider this problem in the context of a generic framework for \emph{mobile reaction networks}, a concise model to describe Markov population processes. This is inspired by stochastic reaction networks, which are frequently used for the analysis of systems of (bio-)chemical reactions (e.g.,~\cite{Gillespie77}).
We assume that nodes may interact with each other within the same region, and undergo unbiased RW with a parameter, the \emph{migration rate}, which can be dependent on their local state. We also consider  \emph{absorbing} and \emph{reflective} boundary conditions for the spatial domain. In the presence of the former, each node that hits the boundary exits the domain without the possibility of ever re-entering it. Instead, if the latter is considered, the boundaries build a barrier of the domain that cannot be crossed by nodes. Those assumptions do not exclude the possibility of having exogenous arrivals into the system, e.g., by means of a Poisson process in each region.

The purpose of this paper is to develop a technique to effectively analyze models when $K$  and the total population of nodes are large. It is natural to consider an approach with two dimensions of scaling: the first one is the celebrated \emph{density-dependent} form that allows us to obtain a continuous deterministic limit for the concentrations (or densities) of the node populations~\cite{kurtz-1970}, while keeping the regions discrete. This is achieved by constructing a family of population processes indexed by a parameter, hereafter denoted by $N$, such that the initial state of each element of the family is a population of nodes that grows linearly with $N$. Then, the result of Kurtz~\cite{kurtz-1970} ensures that the sequence converges, as $N \to \infty$, to the solution of an ODE system of size $L(K + 1)^2$. The second scaling occurs on the spatial dimension. We define a suitable sequence of migration rates such that increasing $K$ means considering regions which are closer and closer to each other on a regular mesh in the unit square in~$\mathbb{R}^2$ (i.e., the regions are at a distance $1/K$ from each other). After setting the scene, we show that the sequence of ODE systems converges, as $K \to \infty$, to a system of $L$ partial differential equations (PDEs) of reaction-diffusion type. In these PDEs, the diffusive terms model the continuous migration across regions, whereas the reactive terms describe the local interactions between nodes. To the best of our knowledge, such a convergence result has not been proven before. We also argue that our limit result allows for a more efficient numerical solution. This may seem surprising because analytical solutions of PDE systems are scarce and numerical solutions of PDE systems rely on a discretization of space. However while in the case of the limit  ODE system of size $L(K + 1)^2$ the discretization is dictated by the stochastic model, in our limit PDE system the coarseness of the discrete mesh depends on the PDE solver. In other words, the PDE solver may in effect give rise to a coarsening of the original spatial domain.

A typical situation of practical interest to which our framework can be applied is that of large-scale mobile networks such as personal communication services (e.g.,~\cite{857925}): there are many base stations (e.g., in a wide-area cellular network) and the area served by a base station can be modeled as a region, which can contain potentially many mobile nodes that may migrate across regions. Here, the modeler may wish to predict how nodes distribute across the network over a given time horizon~\cite{Akyildiz:1996:MLU:234766.234778}. In this paper, we demonstrate the applicability of our method by studying a network of mobile nodes with connectivity provided by an 802.11 access point located in each region. First, we successfully validate our CTMC model against discrete-event simulation using the JiST/SWANS discrete-event simulation framework~\cite{Barr:2005:JEA:1060168.1060170}. Then, we show that the PDE solution provides an excellent estimate of the network's performance even for relatively small population sizes and coarse lattices. We confirm that the discretization of the space domain obtained for the PDE solution induces coarsening; numerically, we observe up to two orders of magnitude speed-ups with respect to the solution of the lattice-dependent limit ODE system.  This promotes our approach as a versatile tool for the accurate evaluation of large-scale mobile systems.

\paragraph*{Paper outline} The remainder of the paper is organized as follows. Section~\ref{sec:related} overviews related work. In order to fix notation and provide intuition on the scaling limits considered in the paper, Section~\ref{sec:srn} introduces \emph{stationary reaction networks}, which describe \emph{static} networks with no mobility, whilst allowing nodes to be described by a local state space. Here we discuss how such models admit a classic fluid limit as a system of ODEs.
Section~\ref{sec:smrn} presents mobile reaction networks, which are defined as a conservative extension of stationary networks with an explicit mobility model. We introduce a straightforward spatial ODE limit that depends upon the lattice granularity. In Section~\ref{sec:pde} we discuss the main contribution of this paper, namely the convergence of the spatial ODE limit to a system of reaction-diffusion PDEs by assuming that nodes undergo unbiased random walk during their evolution. Section~\ref{sec:validation} discusses the numerical tests on our validation model. Finally, Section~\ref{sec:conclusion} concludes the paper. %

\section{Related Work}\label{sec:related}


\paragraph*{Reaction-diffusion PDEs} PDEs of reaction-diffusion type are very well understood in many disciplines, such as biology~\cite{murray:i}, ecology~\cite{Levin01}, and chemistry~\cite{kampen01}. It is beyond the scope of this paper to provide a general overview of the literature. Instead, we focus on related work that, similarly to ours, considers PDEs as the macroscopic deterministic behavior of a stochastic process.

In physical chemistry, one such approach is to consider the so-called \emph{reaction-diffusion master equation}, which corresponds to the forward equations of a CTMC that models a network of biochemical reactions occurring at discrete sites, and molecular transitions across sites~\cite{doi:10.1137/070705039}. Although the stochastic model corresponds to ours, a result of convergence is not provided as done in this paper, for two reasons. First, the underlying ODE system is obtained by means of a moment-closure technique, meaning that~\cite{doi:10.1137/070705039} \emph{makes the approximation} that for a given stochastic process $X(t)$ and a function $f$ it holds that $\mathbb{E}[f(X(t))] \approx f(\mathbb{E}[X(t)])$, although this does not hold in general for nonlinear $f$. Second,~\cite{doi:10.1137/070705039} \emph{assumes} implicitly that the underlying ODE systems of size $\mathcal{O}(K^2)$ converge, as $K \to \infty$, to a PDE system. Instead, we prove both facts formally by applying the following strategy. First, we use Kurtz's result to show that the underlying rescaled stochastic process converges in probability to an ODE system of size $L(K + 1)^2$. In this respect, we justify approximation $\mathbb{E}[f(X(t))] \approx f(\mathbb{E}[X(t)])$ as being asymptotically true in the limit of infinite populations. Afterwards, we use numerical analysis to show that the ODE systems converge, as $K \to \infty$, to a PDE system of size $L$. By combining both statements, we are able to prove that the underlying rescaled stochastic process \emph{converges in probability} to a PDE system of size $L$.

In the recent work~\cite{valuetools14mtmt} the authors study a spatial process algebra for the modeling of mobile systems by means of reaction-diffusion PDEs. However, similarly to~\cite{doi:10.1137/070705039}, no formal proof of convergence is provided, meaning that the ODE system is \emph{assumed to converge} to the reaction-diffusion PDE system. Moreover,~\cite{valuetools14mtmt} is not as expressive as our framework, because it considers a specific communication pattern and supports only reflective boundaries.

The closest approach to ours is found in~\cite{Arnold1980}, which was subsequently worked upon in a series of papers~\cite{Kotelenez1986,Blount1991,Blount1993,Blount1996}.
In the case of absorbing boundaries, their stochastic models are different in that the local reaction rates model only birth and death, i.e. interactions between multiple types of agents are not possible. In the case of reflective boundary conditions, instead, we generalize~\cite{Kotelenez1986} because we can treat more than one type of node.



\paragraph*{Mobility models} There is a substantial amount of work on mobility models, both at the analytical level and  experimentally through traces. In this paragraph we overview the literature that is most closely related to our approach based on the RW model; for an extensive discussion, we refer to the survey~\cite{WCM:WCM72}.

Owing to its analytical tractability, the RW model has been extensively studied in networking research. A discrete-time Markovian model was developed in~\cite{DBLP:conf/infocom/Bar-NoyKS94} for the comparative evaluation of update strategies in cellular networks. The paper considers a mobility model where a node is characterized by states that determine the direction of movement. In this respect our approach is analogous, since agent types may have different movement rates. However, the results of their analysis are presented for a one-dimensional topology (over a ring) and cannot be lifted to more general local interactions; on the other hand, their mobility model is anisotropic. Unbiased RWs are instead proposed in~\cite{Akyildiz:1996:MLU:234766.234778} and~\cite{857925} to study movements across cellular networks, and to study routing protocols~\cite{DBLP:conf/infocom/IoannidisM06,1642727} and performance characteristics in ad-hoc networks~\cite{1354518}. In all cases,  mobility is not coupled with a model of interaction between nodes. RW is also used in~\cite{1498494} as the basic mobility model by which the authors arrive at a reaction-diffusion type equation for information propagation in ad-hoc networks; their analysis is carried out at the deterministic/macroscopic level without considering a limiting regime of a counterpart microscopic/stochastic process. Instead, the PDEs used in~\cite{Garetto:2007:ARM:1313051.1313242} are interpreted as the deterministic limit of the empirical measure of node concentrations, by appealing to the strong law of large numbers.

\paragraph*{Summary} In conclusion, deterministic characterizations of large-scale models of interacting agents are well established in the spatially homogeneous case; on the other hand, diffusion-reaction type equations are well understood when only the spatial/mobility model is of interest, and the results available for cases that also include certain kinds of local interactions do not generalize. In this paper, we bridge the gap between these two perspectives by providing a generic framework for modeling a richer variety of local interactions coupled with a (possibly state-dependent) RW model. 


\section{Stationary Reaction Networks}\label{sec:srn}

In this section we provide a definition of a classic stochastic reaction network which does not consider an explicit mobility model. For this reason, this is called a \emph{stationary reaction network}. For illustrative purposes, we also show how this notation allows us to recover two different fluid models of networked computing systems already studied in the literature. This sets the stage for Section~\ref{sec:smrn}, which, instead, discusses the novel contribution of this paper, consisting of extending a stationary reaction network with space and mobility by means of a RW.

We introduce a CTMC population process where nodes are characterized by $L$ local states, denoted by $A_1, \ldots, A_L$. Our CTMC has state descriptor
\[
\vec{A} := (A_1, \ldots, A_L),
\]
which gives the node populations in each local state. Nodes engage in $J$ \emph{interactions}, labeled $1, \ldots, J$. These are given in the \emph{reaction notation}
\begin{equation}\label{eq:reaction}
j: \quad \sum_{1 \leq l \leq L} c_{jl} A_l \lact{F_j} \sum_{1 \leq l \leq L} d_{jl} A_l, \quad 1 \leq j \leq J.
\end{equation}
The arrow label  $F_j$ denotes the function $F_j : \mathbb{R}^{L + P} \rightarrow \mathbb{R}$ that describes the rate at which the interaction occurs. It is dependent on the population of nodes in each of the $L$ local states, and on $P \geq 0$ rate parameters. Additionally, $F_j$ may explicitly depend on the \emph{scaling parameter} $N \in \mathbb{N}$; this will be used later in this section to construct a sequence of CTMCs, indexed by $N$, that will enjoy convergence to the ODE fluid limit.  The parameters $c_{jl}$ and $d_{jl}$, with $1\leq j \leq J$, $1 \leq l \leq L$ are instead nonnegative integers that describe the number of nodes in each local state that are simultaneously affected by the $j$-th transition. Without loss of generality, we assume that $c_{jl} > 0$ and $0 \leq A_l  < c_{jl}$ imply $F_j(A_1, \ldots, A_L, \beta_1, \ldots, \beta_P\big) = 0$ for all $1 \leq j \leq J$ and $1 \leq l \leq L$. This condition ensures that population do not become negative after a transition occurs, thus maintaining the model meaningful from a physical point of view.
Let a transition rate of the CTMC between states $\vec{A}$ and $\vec{A'}$ be denoted by $q(\vec{A}; \vec{A'})$.  In our stationary reaction network, the rates are given by
\[
q\Big(\vec{A};  (A_1 + d_{j1} - c_{j1}, \ldots, A_L + d_{jL} - c_{jL} ) \Big) = \\
\sum_{j' \in S_j} F_{j'}(A_1, \ldots, A_L, \beta_1, \ldots, \beta_P),
\]
where $S_j$ is the set of all actions which lead to the same state change as the $j$-th action, i.e., $S_j = \big\{ 1 \leq j' \leq J \mid \forall 1 \leq l \leq L ( d_{j'l} - c_{j'l} = d_{jl} - c_{jl} ) \big\}$.

Let us now show two examples of how models that are already available in the literature can be captured with this definition.
\begin{example}
We consider one of the models of epidemic routing studied in~\cite{Zhang20072867}. Its basic version has a single local state, $A_1$, which denotes the number of susceptible nodes within an overall population of $N$ nodes (thus, $N - A_1$ is the population of infected nodes).
The single reaction is
\[
1: \qquad 1 A_1 \xrightarrow{F_1} 0 A_1, \qquad F_1 = \frac{1}{N} \beta A_1 (N - A_1),
\]
where $\beta > 0$ is the contact rate for the infection. For instance, assuming $N > 5$, the state $\vec{A} = (5)$ has a transition to $\vec{A'} = (5 + 0 - 1) = (4)$ with rate $5 \beta (1 - 5/N)$.
\end{example}

\begin{example}\label{ex:bit-torrent}
The deterministic model of peer-to-peer file sharing presented in~\cite{Srikant04} can be shown to be the fluid limit of the following CTMC population process. Let $A_1$ denote the number of \emph{downloaders}, i.e., peers that have a partial copy of a shared file, and let $A_2$ denote the number of \emph{seeds}, i.e., peers that have already completed the download. The reactions are:
\begin{align*}
1: \ 0 A_1 + 0 A_2 & \xrightarrow{F_1} 1 A_1 + 0 A_2,  &  F_1 & = \lambda N, \\
2: \ 1 A_1 + 0 A_2 & \xrightarrow{F_2} 0 A_1 + 0 A_2 , &  F_2 & = \sigma A_1, \\
3: \ 0 A_1 + 1 A_2 & \xrightarrow{F_3} 0 A_1 + 0 A_2 , &  F_3 & = \gamma A_2, \\
4: \ 1 A_1 + 1 A_2 & \xrightarrow{F_4} 0 A_1 + 2 A_2, & F_4 & = \min\{c A_1, \mu(\eta A_1 + A_2)\}.
\end{align*}
Interaction 1 gives the rate of arrivals of new downloaders into the network, parametrized by $\lambda > 0$ and dependent on the network size via $N$. Interactions 2 and 3 give the rate of exit of downloaders and seeds, respectively, with $\sigma, \gamma > 0$ being their individual exit rates. Finally interaction 4 gives the rate at which a downloader completes the file, thus becoming a seed. This is based on a bandwidth-sharing argument between the total requested capacity by downloaders, $cA_1$, with  $c > 0$, and the total upload capacity of the system. The latter is given by the total upload seed capacity (where $\mu > 0$ is the individual upload capacity) and that of the downloaders who are sharing their  partial copy; $0 \leq \eta \leq 1$ is the probability that a requested portion of the file is available at a peer.
\end{example}

The following model represents the stationary version of the case study used in Section~\ref{sec:validation} for the numerical validation.
\begin{example}[Nodes with on/off behavior]\label{ex:case-study}
Let us consider a node with two local states, $A_1$, and $A_2$, denoting an \emph{off} state where the node does not use network capacity, and an \emph{on} state where it downloads a file, respectively. Let us assume that all nodes share the same capacity $c N$ and that $1/\lambda$ is the average time spent by the node in the off state. The model is given by the following reactions.
\begin{align*}
1: \quad 1 A_1 + 0 A_2 & \xrightarrow{F_1} 0 A_1 + 1 A_2, & F_1 & = \lambda A_1, \\
2: \quad 0 A_1 + 1 A_2 & \xrightarrow{F_2} 1 A_1 + 0 A_2, & F_2 & = c \min\{A_2, N\}.
\end{align*}
The interpretation of the interaction function $F_2$ is similar to $F_4$ in Example~\ref{ex:bit-torrent} and is based on the same argument of bandwidth sharing. Specifically, $1/c$ is the average time to download a file. The network offers a capacity equal to $cN$, whereas the downloading users request an overall capacity $c A_2$. Thus, the actual number of files downloaded per unit time is $\min\{ cA_2, cN\} = F_2$. Notice that this does introduce resource contention, hence delay experienced by the downloading nodes. For instance, for $N=1$ and $A_2 > 1$ the capacity is $c$, which will be shared amongst all nodes. 
 \end{example}

In order to obtain a fluid limit for a stationary reaction network, we define a CTMC sequence,  denoted by $(\vec{A}_N(t))_{t \geq 0}$ and indexed by the scaling parameter $N$ such that the initial state of the $N$-th CTMC is given by $A_l(0) = \lfloor N \alpha_l^0 \rfloor$, where $\alpha_l^0 \in \mathbb{R}_{\geq 0}$. Therefore the initial population of nodes increases linearly with $N$. If one considers the rescaled CTMC sequence $(\vec{A}_N(t)/N)_{t \geq 0}$ in the examples above, the larger $N$ the faster the rates and, at the same time, the smaller the jumps; these are of order $O(1/N)$ because the non-normalized CTMC has unitary decreases in the populations. Informally, this behavior suggests a trend that is continuous in the limit.

More formally, convergence of $(\vec{A}_N(t)/N)_{t \geq 0}$ to such a fluid limit needs two assumptions. First, we require that for each $1 \leq j \leq J$ there exist continuous functions $f_j : \mathbb{R}^{L + P} \rightarrow \mathbb{R}$ and $g_j : \mathbb{R}^{L + P} \rightarrow \mathbb{R}_{\geq0}$ such that
\[
\textbf{\bfseries (A1):} \quad \frac{1}{N} F_j\big(N a_1, \ldots, N a_L, b_1, \ldots, b_P\big)
= f_j\big(a_1, \ldots, a_L, b_1, \ldots, b_P\big) + \mathcal{O}\left(\frac{g_j(a_1,\ldots,a_L,b_1,\ldots,b_P)}{N}\right) .
\]

Second, we require that for all $\vec{z}_0 \in \mathbb{R}^{L + P}$ there exist an open neighborhood $\mathfrak{O}$ of $\vec{z}_0$ and a $C \in \mathbb{R}_{\geq0}$ such that
\begin{equation*}
\textbf{\bfseries (A2):} \quad \big|f_j\big(\vec{z}_2) - f_j\big(\vec{z}_1)\big| \leq C \| \vec{z}_2 - \vec{z}_1 \| , \qquad \forall \vec{z}_1, \vec{z}_2 \in \mathfrak{O} .
\end{equation*}

In essence, \textbf{(A1)} asserts that the underlying CTMC is in the density-dependent form, while \textbf{(A2)} requires that each $f_j$ be \emph{locally} Lipschitz continuous. Both conditions will allow us to formally relate the CTMC sequence to a system of $L$ ODEs. For instance,
with regard to Example~\ref{ex:case-study} we have that
\[
f_1(a_1, a_2) = \lambda a_1 \qquad \text{and} \qquad f_2(a_1, a_2) = c \min \{ a_2, 1\}
\]
satisfy \textbf{\bfseries (A1)} for $F_1$ and $F_2$, respectively, and enjoy locally Lipschitz continuity. These functions constitute the vector field of the ODE limit. More precisely, the following holds.

\begin{theorem}[ODE Fluid Limit of Stationary Networks]\label{thm_kurtz_stat}
Let us define
\begin{equation}\label{eq:ode.term.stationary}
\mathcal{L}_l(\vec{a}) = \sum_{1 \leq j \leq J} (d_{jl} - c_{jl}) f_j\big(a_1, \ldots, a_L,\beta_1, \ldots, \beta_P\big) .
\end{equation}
Then, the ODE system
\begin{equation*}
\frac{d}{dt} a_l(t) = \mathcal{L}_l(\vec{a}(t))
\end{equation*}
subjected to the initial condition $\vec{\alpha}^0 = (\alpha_1^0, \ldots, \alpha_L^0)$ has a unique solution $\vec{a}$ in $\mathbb{R}^L$. Moreover, under the assumption that for an arbitrary but fixed $T > 0$ the time domain of $\vec{a}$ contains $[0;T]$, it holds that
\begin{equation*}
\lim_{N \rightarrow \infty} \mathbb{P} \left\{ \sup_{0 \leq t \leq T} \left\| \frac{1}{N}A_N(t) - \vec{a}(t)\right\|_\infty > \varepsilon \right\} = 0, \quad \forall \varepsilon > 0,
\end{equation*}
where the supremum norm $\| \cdot \|_\infty$ on $\mathbb{R}^L$ is given by $\| \vec{a} \|_\infty := \max_{1 \leq l \leq L} | a_l |$.
\end{theorem}
We do not give the proof of this theorem explicitly because it can be seen as a special case of the \emph{spatial ODE fluid limit}, see Theorem~\ref{thm_kurtz} in Section~\ref{sec:smrn}.

For instance, the ODEs for Example~\ref{ex:case-study} are
\begin{equation}\label{ex:ode}
\begin{split}
\frac{d}{dt} a_1(t) & = - \lambda a_1(t) + c \min \{ a_2,1 \}, \\
\frac{d}{dt} a_2(t) & = + \lambda a_1(t) - c \min \{ a_2,1 \}. \\
\end{split}
\end{equation}
Let us consider $\vec{\alpha}^0 = (\alpha_1^0, \alpha_2^0) = (1/4,3/4)$, i.e., we study the model under the conditions where there are initially 25\% of nodes in local state $A_1$ and 75\% of nodes in local state $A_2$. This leads to setting $a_1(0) = 1/4$ and $a_2(0) = 3/4$. The initial states in the  CTMC sequence $\big(\vec{A}_N(t)\big)_{t\geq0}$ are given by
\[
A_1(0) = \left\lfloor \frac{N}{4} \right\rfloor \ \ \text{and} \ \  A_2(0) = \left\lfloor \frac{3N}{4} \right\rfloor, \qquad N \geq 1 .
\]
Let us note that $\frac{1}{N} \vec{A}_N(0) \to (1/4,3/4)$, as $N \to \infty$, i.e., the initial populations of the rescaled Markov population processes converge to the densities given as the ODE initial condition.

\section{Mobile Reaction Networks}\label{sec:smrn}

A mobile reaction network is a Markov population process with an explicit notion of locality and mobility. Space is partitioned in a number of regions. Nodes within the same region may communicate with each other using the interaction functions (\ref{eq:reaction}). Additionally, nodes may move to neighboring regions by performing an unbiased RW. The purpose of this section is to show that a straightforward spatial ODE limit result for such a CTMC where the explicit location of nodes must be kept track of leads to a number of ODEs that grows with the number of regions, hindering the practical applicability of the analysis for large-scale mobile systems. This is an intermediate deterministic model that will be used in Section~\ref{sec:pde} to study a suitable continuous approximation of space by means of a system of $L$ PDEs.

\paragraph*{Space Model} Our space model consists of a lattice in the unit square with $(K+1)^2$ regions, denoted by $\mathcal{R}_K := \{ (i \ds,j \ds) \mid 0 \leq i,j \leq K \}$, where $\ds := 1/K$. To enhance readability, the boundary of the lattice will be denoted by $\Omega_K := \big\{ (x,y) \mid x \in \{0,1\} \lor y \in \{0,1\} \big\}$. A node in the interior $\mathcal{R}_K \setminus \Omega_K$ can travel to one of its neighboring regions $(x-\ds,y)$, $(x+\ds,y)$, $(x,y-\ds)$ and $(x,y+\ds)$. In the case where absorbing boundary conditions are assumed, nodes that migrate to the boundary $\Omega_K$ disappear. Instead, if reflective boundary conditions are in place, nodes at the boundary $\Omega_K$ can travel only to two or three regions. In the following, $\mathcal{N}(x,y)$ denotes the neighboring regions of $(x,y) \in \mathcal{R}_K$, e.g. $\mathcal{N}(0,0) = \{(\ds,0),(0,\ds)\}$.

\subsection{Absorbing Boundary Conditions}

\paragraph*{Stochastic Model} The state descriptor of our population process is
\[
\vec{A} := (A_1^{(x,y)}, \ldots, A_L^{(x,y)})_{(x,y) \in \mathcal{R}_K},
\]
which gives the agent populations in each local state at each region. This conservatively extends the stationary model in Section~\ref{sec:srn} by keeping track of the population of nodes in each of the (discrete) regions in the space domain.

Let us now proceed with the discussion of the transition rates for the RW and for the local interactions separately.

For the RW, we let $\mu_l^K \geq 0$ denote the migration rate for nodes of type-$l$ in a region contained in $\mathcal{R}_K$. The transition rates for an unbiased RW are defined as
\[
q\big(\vec{A} ; (\ldots, A_l^{(x,y)} - 1, A_l^{(\tilde{x},\tilde{y})} + \mathds{1}(\tilde{x},\tilde{y})) \big) =
\begin{cases}
\mu^K_l A_l^{(x,y)} & , \ (\tilde{x},\tilde{y}) \in \mathcal{R}_K \setminus \Omega_K, \\
| \mathcal{N}(x,y) \cap \Omega_K | \ \mu_l^K A_l^{(x,y)} & , \ \text{otherwise,}
\end{cases}
\]
where $(x,y) \in \mathcal{R}_K \setminus \Omega_K$, $(\tilde{x}, \tilde{y}) \in \mathcal{N}(x,y)$. The ellipsis indicates that, apart from the elements of the state vector that are explicitly written, no other elements are affected by the transition. The function $\mathds{1}(\cdot) := \mathds{1}_{\mathcal{R}_K \setminus \Omega_K}(\cdot)$ abbreviates the indicator of the set of inner regions $\mathcal{R}_K \setminus \Omega_K$.
This formally describes an unbiased RW since all the allowed neighboring regions are visited with equal probabilities. The boundary regions, instead, absorb any node that visits them.

The transition rates due to the interaction functions are similar to~(\ref{eq:reaction}). However, in our mobile reaction network we allow the rate parameters---but not the scaling parameter $N \in \mathbb{N}$---to be possibly dependent on the region. This might be useful, for instance, to model inhomogeneous server capacities across the spatial domain. To capture this, for any $1 \leq j \leq J$, $F_j$ is a function $F_j : \mathbb{R}^{L + P} \rightarrow \mathbb{R}$ where the $P$ parameters are given by continuous functions
$$\beta_p : [0;1]^2 \rightarrow \mathbb{R}_{\geq0},  \quad \text{with~} 1 \leq p \leq P.$$
The transition rates are given by
\[
q\Big(\vec{A}; \big( \ldots, A_1^{(x,y)} + d_{j1} - c_{j1}, \ldots, A_L^{(x,y)} + d_{jL} - c_{jL} \big) \Big) =
\sum_{j' \in S_j} F_{j'}\Big(A^{(x,y)}_1, \ldots, A^{(x,y)}_L, \beta_1(x,y), \ldots, \beta_P(x,y)\Big),
\]
where $(x,y) \in \mathcal{R}_K \setminus \Omega_K$. We notice that, indeed, this definition induces \emph{local interactions}, i.e., only the population levels related to the same region $(x,y)$ are affected by the transition. Thus, overall, our model consists of local communication between nodes in a region, together with a mobility model which can be dependent on the local state of the node.

\begin{figure}[t]
\begin{subfigure}{0.49\linewidth}
\centering
\includegraphics[width=4cm]{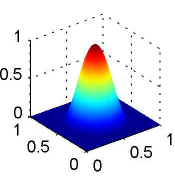}
\caption{}\label{fig:theta}
\end{subfigure}
\
\begin{subfigure}{0.49\linewidth}
\centering
\includegraphics[width=4cm]{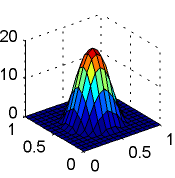}
\caption{}\label{fig:theta.dis}
\end{subfigure}
\caption{(a) Plot of the function $\theta(x,y)$; (b) Initial CTMC conditions from $\theta(x,y)$ for $K = 16$ and $N = 20$.}\label{fig_pde_0}
\end{figure}

\paragraph*{Spatial ODE Fluid Limit} We now consider a CTMC sequence $(\vec{A}_N(t))_{t \geq 0}$ indexed by $N$. The initial state of the $N$-th CTMC is given by $A^{(x,y)}_l(0) = \lfloor N \alpha_l^0(x,y) \rfloor$, where $\alpha_1^0, \ldots, \alpha_L^0 : [0;1]^2 \rightarrow \mathbb{R}_{\geq0}$ are now continuously differentiable functions which are zero at the boundary of $[0;1]^2$. The assumption on the scaling with $N$ is analogous to that of a stationary reaction network. Instead, requiring that the functions are zero at the boundary is for consistency with the fact that the boundary $\Omega_K$ of $\mathcal{R}_K$ does not contain any nodes. Formally, this means that $\alpha^0_l$ have to satisfy the Dirichlet Boundary Conditions (DBCs)
\begin{align}\label{eq:dbc}
\alpha_l^0(x,1) & = 0, & \alpha_l^0(x,0) & = 0, &  \alpha_l^0(1,y) & = 0, & \alpha_l^0(0,y) = 0,
\end{align}
for all $1 \leq l \leq L$, $x \in [0;1]$ and $y \in [0;1]$. For example, in the validation study in Section~\ref{sec:validation} we will consider initial conditions based on a function, denoted by $\theta$, which satisfies the DBCs and is defined as
\begin{align}\label{eq_theta}
\theta(x,y) =
\begin{cases}
0  & \ , \lVert (x,y) - (\frac{1}{2},\frac{1}{2}) \rVert ^2 \geq \frac{1}{4} \\
\frac{\text{exp}(4)}{\text{exp}\left( 1 /(\frac{1}{4} - \lVert (x,y) - (\frac{1}{2},\frac{1}{2}) \rVert ^2 )\right)}  & \ , \text{otherwise.}
\end{cases}
\end{align}
This function is plotted in Figure~\ref{fig:theta}. To show how to obtain suitable initial conditions for the CTMC that yield a spatial ODE fluid limit, let us consider the case $K = 16$ and $N = 20$ and assume that the initial CTMC conditions for some $l$ are given by $A^{(x,y)}_l(0) = \lfloor N \theta(x,y) \rfloor$, for all $(x,y) \in \mathcal{R}_K$. Then, these are represented by the values taken at the mesh visualized in Figure~\ref{fig:theta.dis}. In the limit $K,N \rightarrow \infty$ the normalized initial conditions $\lfloor N \theta(x,y) \rfloor / N$ converge to $\theta(x,y)$.

The fluid limit for a mobile reaction network is obtained similarly to the stationary case. Let us denote $\vec{a} \in \mathbb{R}^{\mathcal{R}_K \times \{1, \ldots, L\}}$ by $(a^{(x,y)}_1, \ldots, a^{(x,y)}_L)_{(x,y) \in \mathcal{R}_K}$. The corresponding system of ODEs is given by
\begin{equation}\label{eq_ode_sys}
\frac{d}{dt} a^{(x,y)}_l(t) = \mathds{1}(x,y) \Big( \mathcal{M}^{(x,y)}_l(\vec{a}(t)) + \mathcal{L}^{(x,y)}_l(\vec{a}(t)) \Big),
\end{equation}
where $(x,y) \in \mathcal{R}_K$, and, for all $1 \leq l \leq L$,
\begin{equation}\label{eq:ode.terms}
\begin{split}
\mathcal{M}^{(x,y)}_l(\vec{a}) & = \mu^K_l \sum_{(\tilde{x},\tilde{y}) \in \mathcal{N}(x,y)} (a_l^{(\tilde{x},\tilde{y})} - a_l^{(x,y)}), \\
\mathcal{L}^{(x,y)}_l(\vec{a}) & = \sum_{1 \leq j \leq J} (d_{jl} - c_{jl})
	f_j\big(a^{(x,y)}_1, \ldots, a^{(x,y)}_L, \beta_1(x,y), \ldots, \beta_P(x,y)\big).
\end{split}
\end{equation}
Essentially, $\mathcal{M}_l^{(x,y)}(\cdot)$ gives the \emph{diffusive} dynamics of the system, i.e., the contribution to the vector field due to the movement of type-$l$ nodes in each region $(x,y)$;
$\mathcal{L}_l^{(x,y)}(\cdot)$, instead, encodes the \emph{reactive} dynamics, by considering the effects on each variable of the local interactions $f_j$. The latter essentially corresponds to the vector field of the corresponding stationary reaction network, i.e., compare $\mathcal{L}^{(x,y)}_l(\cdot)$ with $\mathcal{L}_l(\cdot)$ in (\ref{eq:ode.term.stationary}). For instance, the ODEs of Example~\ref{ex:case-study} in $\mathcal{R}_K$ are
\begin{equation}\label{eq_ode_re_spatial}
\begin{split}
\frac{d}{dt} a^{(x,y)}_1(t) & = \mathds{1}(x,y) \Big( - \lambda a_1^{(x,y)}(t) + c \min \{ a^{(x,y)}_2,1 \}
+ \mu^K_1 \sum_{(\tilde{x},\tilde{y}) \in \mathcal{N}(x,y)} (a_1^{(\tilde{x},\tilde{y})} - a_1^{(x,y)}) \Big), \\
\frac{d}{dt} a^{(x,y)}_2(t) & = \mathds{1}(x,y) \Big( + \lambda a_1^{(x,y)}(t) - c \min \{ a^{(x,y)}_2,1 \}
+ \mu^K_2 \sum_{(\tilde{x},\tilde{y}) \in \mathcal{N}(x,y)} (a_2^{(\tilde{x},\tilde{y})} - a_2^{(x,y)}) \Big) . \\
\end{split}
\end{equation}
In each ODE, the first two terms in the right-hand side correspond to those in (\ref{ex:ode}). The second lines in each equation, instead, describe the spatial behavior of the nodes.


The following result formalizes the convergence to the spatial ODE fluid limit.

\begin{theorem}[ODE Fluid Limit of Spatial Networks]\label{thm_kurtz}
The ODE system (\ref{eq_ode_sys}) subjected to the initial condition
\[
a_l^{(x,y)}(0) = \alpha_l^0(x,y), \qquad (x,y) \in \mathcal{R}_K, \quad 1 \leq l \leq L,
\]
has a unique solution $\vec{a}$ in $\mathbb{R}^{\mathcal{R}_K \times \{1, \ldots, L\}}$. Moreover, under the assumption that for an arbitrary but fixed $T > 0$ the time domain of $\vec{a}$ contains $[0;T]$, it holds that
\begin{equation*}
\lim_{N \rightarrow \infty} \mathbb{P} \left\{ \sup_{0 \leq t \leq T} \left\| \frac{1}{N}A_N(t) - \vec{a}(t)\right\|_\infty > \varepsilon \right\} = 0, \quad \forall \varepsilon > 0,
\end{equation*}
where the supremum norm $\| \cdot \|_\infty$ on $\mathbb{R}^{\mathcal{R}_K \times \{1, \ldots, L\}}$ is given by $\| \vec{a} \|_\infty := \max_{(x,y) \in \mathcal{R}_K, 1 \leq l \leq L} | a_l^{(x,y)} |$.
\end{theorem}

This result can be recovered from Theorem 2.11 in \cite{kurtz-1970}, by a rewrite according to the notation of reaction networks adopted in this paper. However, this version represents a slight modification of that theorem in that we relax the assumption of \emph{global} Lipschitz continuity to \emph{local} Lipschitz continuity given by {\bfseries (A2)}, by assuming that the time domain of the unique ODE solution contains $[0;T]$. It is restated in this way because it is this version that will be used in the proof of our main result, Theorem \ref{thm_main}.

\subsection{Reflective Boundary Conditions}

\paragraph*{Stochastic Model} The case of reflective boundaries differs from that of absorbing boundaries only on the boundary $\Omega_K$. In particular, the migration rates are given by $q\big(\vec{A} ; (\ldots, A_l^{(x,y)} - 1, A_l^{(\tilde{x},\tilde{y})} + 1) \big) = \mu^K_l A_l^{(x,y)}$, where $(x,y) \in \mathcal{R}_K$, $(\tilde{x},\tilde{y}) \in \mathcal{N}(x,y)$ and $1 \leq l \leq L$. The reaction rates, instead, are
\[
q\Big(\vec{A}; \big( \ldots, A_1^{(x,y)} + d_{j1} - c_{j1}, \ldots, A_L^{(x,y)} + d_{jL} - c_{jL} \big) \Big) =
\sum_{j' \in S_j} F_{j'}\Big(A^{(x,y)}_1, \ldots, A^{(x,y)}_L, \beta_1(x,y), \ldots, \beta_P(x,y)\Big),
\]
where $(x,y) \in \mathcal{R}_K$.

\paragraph*{Spatial ODE Fluid Limit} The initial states of the underlying CTMC sequence $(\vec{A}_N(t))_{t \geq 0}$ are given by $A^{(x,y)}_l(0) = \lfloor N \alpha_l^0(x,y) \rfloor$, where $\alpha_1^0, \ldots, \alpha_L^0 : [0;1]^2 \rightarrow \mathbb{R}_{\geq0}$ are now continuously differentiable functions whose derivatives are zero at the boundary, that is
\begin{align}\label{eq:nbc}
0 & = \partial_x \alpha_l^0(1,y) &
0 & = \partial_x \alpha_l^0(0,y) &
0 & = \partial_y \alpha_l^0(x,1) &
0 & = \partial_y \alpha_l^0(x,0)
\end{align}
for all $x,y \in [0;1]$ and $1 \leq l \leq L$. The conditions (\ref{eq:nbc}) are formally known as (zero) Neumann boundary conditions, short NBCs. It can be shown that (\ref{eq_theta}) satisfies, apart from the DBCs (\ref{eq:dbc}), also the NBCs (\ref{eq:nbc}).

The fluid limit for a mobile reaction network is given by
\begin{equation}\label{eq_ode_sys_nbcs}
\frac{d}{dt} a^{(x,y)}_l(t) = \mathcal{M}^{(x,y)}_l(\vec{a}(t)) + \mathcal{L}^{(x,y)}_l(\vec{a}(t)) ,
\end{equation}
where $(x,y) \in \mathcal{R}_K$ and $\mathcal{M}^{(x,y)}_l$, $\mathcal{L}^{(x,y)}_l$ are as in (\ref{eq:ode.terms}). In particular, in the case of reflective boundary conditions, the ODEs of Example~\ref{ex:case-study} in $\mathcal{R}_K$ are
\begin{equation*}
\begin{split}
\frac{d}{dt} a^{(x,y)}_1(t) & = - \lambda a_1^{(x,y)}(t) + c \min \{ a^{(x,y)}_2,1 \}
+ \mu^K_1 \sum_{(\tilde{x},\tilde{y}) \in \mathcal{N}(x,y)} (a_1^{(\tilde{x},\tilde{y})} - a_1^{(x,y)}) , \\
\frac{d}{dt} a^{(x,y)}_2(t) & = + \lambda a_1^{(x,y)}(t) - c \min \{ a^{(x,y)}_2,1 \}
+ \mu^K_2 \sum_{(\tilde{x},\tilde{y}) \in \mathcal{N}(x,y)} (a_2^{(\tilde{x},\tilde{y})} - a_2^{(x,y)}) . \\
\end{split}
\end{equation*}

The following result formalizes the convergence to the spatial ODE fluid limit.

\begin{theorem}[ODE Fluid Limit of Spatial Networks]\label{thm_kurtz_nbcs}
The ODE system (\ref{eq_ode_sys_nbcs}) subjected to the initial condition
\[
a_l^{(x,y)}(0) = \alpha_l^0(x,y), \qquad (x,y) \in \mathcal{R}_K, \quad 1 \leq l \leq L,
\]
has a unique solution $\vec{a}$ in $\mathbb{R}^{\mathcal{R}_K \times \{1, \ldots, L\}}$. Moreover, under the assumption that for an arbitrary but fixed $T > 0$ the time domain of $\vec{a}$ contains $[0;T]$, it holds that
\begin{equation*}
\lim_{N \rightarrow \infty} \mathbb{P} \left\{ \sup_{0 \leq t \leq T} \left\| \frac{1}{N}A_N(t) - \vec{a}(t)\right\|_\infty > \varepsilon \right\} = 0, \quad \forall \varepsilon > 0 .
\end{equation*}
\end{theorem}


\section{Spatial Reaction-Diffusion Limit}\label{sec:pde}

The spatial ODE fluid limit of Theorem~\ref{thm_kurtz} holds for any arbitrary but fixed $K$. However, for large $K$, the analysis may become infeasible because the ODE system size grows with $L(K+1)^2$. The purpose of this section is to consider a limit behavior that allows the analysis to be independent from $K$. To do so, we study a limit PDE of reaction-diffusion type.

\paragraph*{Space Scaling} In addition to a suitable scaling of the population sizes with $N$, we require a proper scaling of the migration rates with $K$. Specifically, for all $1 \leq l \leq L$, it must hold that
\begin{equation*}
\textbf{\bfseries (A3):} \quad \mu_l^K = \mu_l K^2.
\end{equation*}

Let us motivate this scaling of space by analyzing a continuous-time unbiased RW, denoted by $(W_k(t))_{t\geq0}$, in two dimensions $\frac{1}{k} \mathbb{Z} \times \frac{1}{k} \mathbb{Z}$, with step size $1/k$ and sojourn time at each state that is exponentially distributed with mean $1/r_k$. The following result holds.
\begin{proposition}\label{prop_scaling}
Let us assume that $W_k(0) = (0,0)$ and denote by $d_k(t) := \| W_k(t) \|$ the Euclidian distance of $W_k(t)$ from the origin after time $t$. Then, it holds that $\mathbb{E}(d_k(t)^2) = \left(\frac{r_k}{k^2}\right) t$ for all $k \geq 1$ and $t \geq 0$.
\end{proposition}

We notice that setting $r_k = k^2 r_1$ for all $k \geq 1$ implies that
$\mathbb{E}(d_k(t)^2) = \mathbb{E}(d_1(t)^2)$ for all $k \geq 1$ and for all $t \geq 0$.
This relation states that if the step size of the agent is $1/k$, then its migration rate should be equal to $k^2 r_1$, in order for the RW to always cover \emph{the same distance} on average independently of $k$. Therefore, {\bfseries (A3)} ensures that the average distance covered by a node across the unit square is invariant with respect to the lattice granularity.

\subsection{Absorbing Boundary Conditions}

\paragraph*{Illustrating Example} The spatial PDE limit is obtained by a suitable manipulation of the diffusive dynamics of the ODE system. In order to give some intuition on this procedure, let us now consider an example of a purely diffusive population of nodes of the same type. Its CTMC population process is shown to converge to the celebrated \emph{heat equation} (e.g,~\cite{okubo80}).
\begin{example}\label{ex:heat}
A purely diffusive process can be modeled with a mobile reaction network with $L = 1$ and $J=1$, where the interaction function is given by:
\[
1: \quad 1A_1 \xrightarrow{F_1} 1 A_1,
\]
for any $F_1 \equiv 0$.
\end{example}
The above example indeed gives a purely diffusive process, essentially because there are no stationary, local  transitions. Let us point out that this somewhat degenerate mobile reaction network only serves an illustrative purpose.

Using the results presented in the previous section, for any $K$, the rescaled CTMC \emph{density process} $(\frac{1}{N} \vec{A}_N(t))_{t \geq 0}$ for Example~\ref{ex:heat} converges in probability, as $N \rightarrow \infty$, to the solution of the system of $(K+1)^2$ ODEs
\[
\frac{d}{dt} a_1^{(x,y)}(t) = \mathds{1}(x,y) \mu^K_1 \sum_{(\tilde{x},\tilde{y}) \in \mathcal{N}(x,y)} \big(  a_1^{(\tilde{x},\tilde{y})}(t) - a_1^{(x,y)}(t) \big),
\]
where $a_1^{(x,y)}(0) = \alpha_1^0(x,y)$ for all $(x,y) \in \mathcal{R}_K$, where $\alpha_1^0(x,y)$ is any function that satisfies the DBCs (\ref{eq:dbc}).

We continue by observing that
\begin{align}\label{eq_first_ex_ode_dbcs}
\frac{d}{dt} a_1^{(x,y)}(t) & = \mathds{1}(x,y) \frac{\mu^K_1}{K^2} \lapl a_1^{(x,y)}(t) , &
\lapl a_1^{(x,y)}(t) & := \sum_{(\tilde{x},\tilde{y}) \in \mathcal{N}(x,y)} \frac{a_1^{(\tilde{x},\tilde{y})}(t) - a_1^{(x,y)}(t)}{(1/K)^2},
\end{align}
where $(x,y) \in \mathcal{R}_K$ and $\lapl a_1^{(x,y)}$ refers to the discrete Laplace operator in a two dimensional lattice with neighboring distance $\ds = 1/K$, if $(x,y) \in \mathcal{R}_K \setminus \Omega_K$. In the case where $(x,y) \in \Omega_K$, instead, $\lapl a_1^{(x,y)}$ is not the discrete Laplace operator, since $|\mathcal{N}(x,y)| < 4$.
Using \textbf{(A3)} by setting $\mu^K_1 := \mu_1K^2$, one can simplify (\ref{eq_first_ex_ode_dbcs}) to the ODE system
\begin{equation}\label{eq_first_ode}
\begin{split}
\frac{d}{dt} a_1^{(x,y)}(t) = \mathds{1}(x,y) \mu_1 \lapl a_1^{(x,y)}(t), \quad (x,y) \in \mathcal{R}_K.
\end{split}
\end{equation}

The general result presented in the next section allows us to conclude that the solution to this ODE system converges, as $K \rightarrow \infty$, to that of the heat equation
\begin{equation}\label{eq_first_pde}
\partial_t \alpha_1 = \mu_1 \bigtriangleup \alpha_1 , \ (x,y,t) \in [0;1]^2 \times \mathbb{R}_{\geq0},
\end{equation}
where $\bigtriangleup$ denotes the continuous Laplace operator, subjected to the initial condition $\alpha_1^0$ and the DBCs
\begin{align*}
\alpha_1(1,y,t) & = 0, & \alpha_1(0,y,t) & = 0, & \alpha_1(x,1,t) & = 0, &  \alpha_1(x,0,t) = 0
\end{align*}
for all $t \geq 0$, $x \in [0;1]$ and $y \in [0;1]$.

Thus, we are dividing the question of establishing the convergence of the CTMCs $(\frac{1}{N} \vec{A}_N(t))_{t \geq 0}$ to the solution of the PDE~(\ref{eq_first_pde}) as $K, N \rightarrow \infty$ in two sub-problems where the two limits are studied separately: the first sub-problem involves convergence to a fluid limit with discrete regions, in the classical sense of Kurtz by sending $N$ to infinity; after having established such a limit behavior, the second sub-problem studies the convergence of the spatial ODE limit~(\ref{eq_first_ode}) to the PDE~(\ref{eq_first_pde}), by sending $K$ to infinity.

We wish to provide an insight into the proof technique by noticing that the latter issue only involves deterministic quantities. Hence, it is possible to use classical results from numerical analysis in order to study the convergence. To illustrate our strategy, we start by writing an approximating sequence, denoted by $(\vec{a}(m))_{m\geq0}$, to solve the spatial ODE limit~(\ref{eq_first_ode}) with the Euler method using  a fixed step size $\dt$. The iterations are given by
\begin{equation}\label{eq_first_euler_dbcs}
\begin{split}
a_1^{(x,y)}(0) & := \alpha_1^0(x,y), \\
a_1^{(x,y)}(m + 1) & := a_1^{(x,y)}(m) + \mathds{1}(x,y) \dt \big( \mu_1 \lapl a_1^{(x,y)}(m) \big) ,
\end{split}
\end{equation}
with $(x,y) \in \mathcal{R}_K$ and  $m \geq 0$. It is well-known that the sequence converges to the solution of the ODE system if $\dt \rightarrow 0$~\cite{Gear1971}.

The crucial observation is that  this very sequence can be interpreted as a finite difference scheme (e.g., \cite{thomas95,strikwerda89}) for solving the PDE~(\ref{eq_first_pde}).
This can be seen by discretizing the domain $[0;1]^2 \times \mathbb{R}_{\geq0}$ into $\mathcal{R}_K \times \{ m \dt \mid m \geq 0 \}$ and observing that a sufficiently smooth solution of (\ref{eq_first_pde}) satisfies
\[
\frac{\alpha_1(x,y,(m + 1)\dt) - \alpha_1(x,y,m\dt)}{\dt} =
\mu_1 \sum_{(\tilde{x},\tilde{y}) \in \mathcal{N}(x,y)} \frac{\alpha_1(\tilde{x},\tilde{y},m\dt) - \alpha_1(x,y,m\dt)}{\ds^2} \\
+ \mathcal{O}(\dt + \ds^2),
\]
for all $(x,y) \in \mathcal{R}_K \setminus \Omega_K$, cf. \cite{thomas95}. As the DBCs enforce $\alpha_1(x,y,(m+1)\dt) = 0$ for all $(x,y) \in \Omega_K$, the PDE solution of (\ref{eq_first_pde}) can be approximated by the finite difference scheme that computes the sequence $(\vec{\alpha}(m))_{m\geq0}$, where
\begin{align*}\label{eq_first_pde_diff_scheme}
\alpha_1^{(x,y)}(m + 1) = \alpha_1^{(x,y)}(m) + \mathds{1}(x,y) \dt \Big( \mu_1 \sum_{(\tilde{x},\tilde{y}) \in \mathcal{N}(x,y)} \frac{\alpha_1^{(\tilde{x},\tilde{y})}(m) - \alpha_1^{(x,y)}(m)}{\ds^2} \Big)
= \alpha_1^{(x,y)}(m)  + \mathds{1}(x,y) \dt \left( \mu_1 \lapl \alpha_1^{(x,y)}(m) \right),
\end{align*}
which corresponds to (\ref{eq_first_euler_dbcs}) as required.

Therefore, the desired result of convergence of the stochastic process to the PDE limit amounts to  proving that such a finite difference scheme does converge to the solution of (\ref{eq_first_pde}) as $\dt,\ds \rightarrow 0$. The main result of this paper, presented next,  generalizes this convergence to a class of models where nodes also feature local interactions.

\paragraph*{General PDE Limit}
Let us now consider the ODE limit~(\ref{eq_ode_sys}) for a generic mobile reaction network in $\mathcal{R}_K$. Observe that~\textbf{(A3)} yields
\[
\mathcal{M}^{(x,y)}_l(\vec{a}) = \mu_l \left( \sum_{(\tilde{x},\tilde{y}) \in \mathcal{N}(x,y)} \frac{a_l^{(\tilde{x},\tilde{y})} - a_l^{(x,y)}}{(1/K)^2} \right) = \mu_l \lapl a_l^{(x,y)}.
\]
This allows us to rewrite the system of ODEs (\ref{eq_ode_sys}) as
\begin{equation*}
\frac{d}{dt} a^{(x,y)}_l(t) = \mathds{1}(x,y) \Big(\mu_l \lapl a_l^{(x,y)}(t) + \mathcal{L}^{(x,y)}_l(\vec{a}(t)) \Big),
\end{equation*}
where $a^{(x,y)}_l(0) = \alpha_l^0(x,y)$ for $(x,y) \in \mathcal{R}_K$ and $1 \leq l \leq L$.

Similarly to (\ref{eq_first_euler_dbcs}), the approximating sequence $(\vec{a}(m))_{m\geq0}$ arising from the Euler method with a fixed $\dt$ to solve (\ref{eq_ode_sys}) is given by:
\begin{align}\label{eq_ode_euler}
a_l^{(x,y)}(m+1) := a_l^{(x,y)}(m) + \mathds{1}(x,y) \cdot \dt \cdot \Big( \mu_l \lapl a_l^{(x,y)}(m) + \mathcal{L}^{(x,y)}_l\big(\vec{a}(m)\big) \Big), \qquad m \geq 0 .
\end{align}
We now observe that this sequence corresponds to a finite difference scheme for the \emph{reaction-diffusion} PDE system
\begin{align}\label{eq_pde}
\partial_t \alpha_l = \mu_l \bigtriangleup \alpha_l + \sum_{1 \leq j \leq J} (d_{jl} - c_{jl}) f_j(\alpha_1, \ldots, \alpha_L, \beta_1, \ldots, \beta_P) ,
\end{align}
where $(x,y,t) \in [0;1]^2 \times \mathbb{R}_{\geq0}$ and $1 \leq l \leq L$, subjected to the DBCs
\begin{align}\label{eq_pde_dbcs}
\alpha_l(1,y,t) & = 0, & \alpha_l(0,y,t) & = 0, & \alpha_l(x,1,t) & = 0, & \alpha_l(x,0,t) = 0
\end{align}
for all $1 \leq l \leq L$, $0 \leq t \leq T$, $x \in [0;1]$ and $y \in [0;1]$.

Therefore, in order to prove that the ODEs~(\ref{eq_ode_sys}) converge to the PDEs~(\ref{eq_pde}), it remains to show that the approximating sequence (\ref{eq_ode_euler}) interpreted as a finite difference scheme for (\ref{eq_pde}) is convergent with respect to the supremum norm. This is done in the following.
\begin{theorem}\label{thm_scheme}
Assume that for a fixed $T > 0$, the family of functions $(\alpha_1,\ldots,\alpha_L)$ on $[0;1]^2 \times [0;T]$ describes for all $0 < T' \leq T$ the unique solution of (\ref{eq_pde}) subjected to the DBCs (\ref{eq_pde_dbcs}) on $[0;1]^2 \times [0;T']$. Further, assume that
\[ \partial_{txx} \alpha_l, \ \partial_{xxxx} \alpha_l, \ \partial_{yxx} \alpha_l, \ \partial_{tyy} \alpha_l, \ \partial_{xyy} \alpha_l, \ \partial_{yyyy} \alpha_l, \]
where $1 \leq l \leq L$, exist and are continuous. For $M \geq 1$, define $\dt := T / M$, $r_l := \mu_l {\dt}/{\ds^2}$ and assume that $M$ is large enough such that $r_l \leq 1/4$ for all $1 \leq l \leq L$.

Then, there exist $K_0, M_0 \geq 1$ and $C > 0$ such that the Euler method (\ref{eq_ode_euler}) interpreted as a finite difference scheme of (\ref{eq_pde}) subjected to the DBCs (\ref{eq_pde_dbcs}) satisfies
\[ \max_{0 \leq m \leq M} \max_{(x,y) \in \mathcal{R}_K, 1 \leq l \leq L} | \alpha_l(x,y,m\dt) - a^{(x,y)}_l(m) | \leq C(\ds^2 + \dt)  \]
for all $K \geq K_0$, $M \geq M_0$.
\end{theorem}

Theorem \ref{thm_kurtz} and Theorem \ref{thm_scheme} can now be used to show that the CTMC sequence converges in probability to the solution of the PDE system (\ref{eq_pde}) subjected to the DBCs (\ref{eq_pde_dbcs}) on $[0;1]^2 \times [0;T]$. This is stated in the following, which is the most important result of this paper.

\begin{theorem}[Spatial PDE Limit]\label{thm_main}
Assume that for a fixed $T > 0$, the family of functions $(\alpha_1,\ldots,\alpha_L)$ on $[0;1]^2 \times [0;T]$ describes for all $0 < T' \leq T$ the unique solution of (\ref{eq_pde}) subjected to the DBCs (\ref{eq_pde_dbcs}) on $[0;1]^2 \times [0;T']$. Further, assume that
\[ \partial_{txx} \alpha_l, \ \partial_{xxxx} \alpha_l, \ \partial_{yxx} \alpha_l, \ \partial_{tyy} \alpha_l, \ \partial_{xyy} \alpha_l, \ \partial_{yyyy} \alpha_l, \]
where $1 \leq l \leq L$, exist and are continuous. Then, the time domain of the unique solution $\vec{a}$ of (\ref{eq_ode_sys}) contains $[0;T]$ and for each $\varepsilon > 0$ it holds that
\[
\lim_{K \rightarrow \infty} \lim_{N \rightarrow \infty} \mathbb{P}  \Bigg\{ \sup_{0 \leq t \leq T} \max_{(x,y) \in \mathcal{R}_K, 1 \leq l \leq L}
\left| \frac{1}{N}A_l^{(x,y)}(t) - \alpha_l(x,y,t) \right|  > \varepsilon \Bigg\} = 0.
\]
\end{theorem}

For instance, using {\bfseries (A3)} the system (\ref{eq_ode_re_spatial}) representing the spatial ODE limit for Example~\ref{ex:case-study} can be simplified to
\begin{equation}\label{eq_ode_re_spatial_simp}
\begin{split}
\frac{d}{dt} a^{(x,y)}_1(t) & = \mathds{1}(x,y) \Big( \mu_1 \lapl a^{(x,y)}_1(t) - \lambda a_1^{(x,y)}(t) + c \min \{ a^{(x,y)}_2,1 \} \Big), \\
\frac{d}{dt} a^{(x,y)}_2(t) & = \mathds{1}(x,y) \Big( \mu_2 \lapl a^{(x,y)}_2(t) + \lambda a_1^{(x,y)}(t) - c \min \{ a^{(x,y)}_2,1 \} \Big) . \\
\end{split}
\end{equation}

Then, given that the assumptions of Theorem \ref{thm_scheme} are fulfilled for suitable initial concentrations $\alpha^0_1$ and $\alpha^0_2$, the above ODE system converges, as $K \to \infty$, to the solution of the PDE system
\begin{equation}\label{eq_pde_re}
\begin{split}
\partial_t \alpha_1 & = - \lambda \alpha_1 + c \min \{ \alpha_2, 1 \} + \mu_1 \bigtriangleup \alpha_1 , \\
\partial_t \alpha_2 & = + \lambda \alpha_1 - c \min \{ \alpha_2, 1 \} + \mu_2 \bigtriangleup \alpha_2 . \\
\end{split}
\end{equation}
This implies, as stated by Theorem \ref{thm_main}, that the sequence of normalized CTMCs converges to the solution of the PDE system when $K, N \to \infty$.

\subsection{Reflective Boundary Conditions}

\paragraph*{Illustrating Example} Similarly to absorbing boundary conditions, we first convey our approach on the heat equation. As we will see, the presence of NBCs will ask for a different numerical treatment.

Using similar arguments as in the case of absorbing boundary conditions, we infer that the CTMC sequence $(\frac{1}{N} \vec{A}_N(t))_{t \geq 0}$ underlying Example~\ref{ex:heat} converges in probability, as $N \rightarrow \infty$, to the solution of the system of $(K+1)^2$ ODEs
\begin{equation}\label{eq_first_ode_nbcs}
\begin{split}
\frac{d}{dt} a_1^{(x,y)}(t) = \mu_1 \lapl a_1^{(x,y)}(t), \quad (x,y) \in \mathcal{R}_K ,
\end{split}
\end{equation}
where $a_1^{(x,y)}(0) = \alpha_1^0(x,y)$ for all $(x,y) \in \mathcal{R}_K$. Crucially, $\alpha_1^0(x,y)$ is now a function that satisfies the NBCs (\ref{eq:nbc}).

Further, we observe that the spatial ODE system (\ref{eq_first_ode_nbcs}) can be solved by means of the Euler method using a fixed step size $\dt$. The iterations are given by
\begin{equation}\label{eq_first_euler_nbcs}
\begin{split}
a_1^{(x,y)}(0) & := \alpha_1^0(x,y), \\
a_1^{(x,y)}(m + 1) & := a_1^{(x,y)}(m) + \dt \big( \mu_1 \lapl a_1^{(x,y)}(m) \big) ,
\end{split}
\end{equation}
with $(x,y) \in \mathcal{R}_K$ and $m \geq 0$. As before, the crucial observation is that this very sequence can be interpreted as a finite difference scheme for solving the heat equation~(\ref{eq_first_pde}) which is now subject to the NBCs
\begin{align}\label{eq_first_nbcs}
0 & = \partial_x \alpha_1(1,y,t) &
0 & = \partial_x \alpha_1(0,y,t) &
0 & = \partial_y \alpha_1(x,1,t) &
0 & = \partial_y \alpha_1(x,0,t) ,
\end{align}
where $x,y \in [0;1]$ and $0 \leq t \leq T$. To see this, we discretize the domain $[0;1]^2 \times \mathbb{R}_{\geq0}$ into $\mathcal{R}_K \times \{ m \dt \mid m \geq 0 \}$ and observe that a sufficiently smooth solution of the heat equation (\ref{eq_first_pde}) satisfies
\[
\frac{\alpha_1(x,y,(m + 1)\dt) - \alpha_1(x,y,m\dt)}{\dt} =
\mu_1 \sum_{(\tilde{x},\tilde{y}) \in \mathcal{N}(x,y)} \frac{\alpha_1(\tilde{x},\tilde{y},m\dt) - \alpha_1(x,y,m\dt)}{\ds^2} \\
+ \mathcal{O}(\dt + \ds^2),
\]
for all $(x,y) \in \mathcal{R}_K \setminus \Omega_K$. The treatment of the boundary $\Omega_K$, instead, is done by means of so-called \emph{ghost regions} $(-\ds,j\ds), (1+\ds,\ds), (i\ds,-\ds), (i\ds,1+\ds)$ where $0 \leq i,j \leq K$, cf. \cite[Section 1.4]{thomas95}. These allow one to define the discrete Laplace operator also at $\Omega_K$. More importantly, for a given ghost region $(x_0,y_0)$, the NBCs allow us to express $a^{(x_0,y_0)}_1(m)$ in terms of $\{ a^{(x,y)}_1(m) \mid (x,y) \in \mathcal{R}_K \}$. For instance, the assumptions $\partial_x \alpha_1(0,0,m \dt) = 0$ and $\partial_y \alpha_1(0,0,m \dt) = 0$ yield the approximations
\begin{align*}
\frac{a_1^{(0,0)}(m) - a_1^{(-\ds,0)}(m)}{\ds} & = 0 &  \frac{a_1^{(0,0)}(m) - a_1^{(0,-\ds)}(m)}{\ds} & = 0 ,
\end{align*}
which imply, in turn, that
\begin{equation*}
\begin{split}
a^{(0,0)}_1(m + 1) & = a^{(0,0)}_1(m) + \dt \mu_1 \bigg( \frac{a^{(-\ds,0)}_1(m) - 2a^{(0,0)}_1(m) + a^{(\ds,0)}_1(m)}{\ds^2} + \frac{a^{(0,-\ds)}_1(m) - 2a^{(0,0)}_1(m) + a^{(0,\ds)}_1(m)}{\ds^2} \bigg) \\
& = a^{(0,0)}_1(m) + \dt \bigg( \mu_1 \!\!\!\!\!\! \sum_{(\tilde{x},\tilde{y}) \in \{(\ds,0),(0,\ds)\}} \!\!\!\!\!\!\!\! \frac{a^{(\tilde{x},\tilde{y})}_1(m) - a^{(0,0)}_1(m)}{\ds^2} \bigg) \\
& = a^{(0,0)}_1(m) + \dt \mu_1 \lapl a_1^{(0,0)}.
\end{split}
\end{equation*}
Thus, we recover the exact correspondence with (\ref{eq_first_euler_nbcs}) also at the corner $(0,0)$. As a similar calculation can be shown to apply to all regions in $\Omega_K$, we conclude that the Euler sequence (\ref{eq_first_euler_nbcs}) can be interpreted as a finite difference scheme of the heat equation (\ref{eq_first_pde}) that is subject to NBCs (\ref{eq_first_nbcs}).

\paragraph*{General PDE Limit}
We now consider the ODE limit~(\ref{eq_ode_sys}) for a generic mobile reaction network in $\mathcal{R}_K$. Similarly to the case of absorbing boundary conditions, (\ref{eq_ode_sys}) can be rewritten into
\begin{equation*}
\frac{d}{dt} a^{(x,y)}_l(t) = \mu_l \lapl a_l^{(x,y)}(t) + \mathcal{L}^{(x,y)}_l(\vec{a}(t)) ,
\end{equation*}
where $a^{(x,y)}_l(0) = \alpha_l^0(x,y)$ for $(x,y) \in \mathcal{R}_K$ and $1 \leq l \leq L$. The approximating sequence $(\vec{a}(m))_{m\geq0}$ arising from the Euler method with a fixed $\dt$ to solve (\ref{eq_ode_sys}) is given by
\begin{align}\label{eq_ode_euler_nbcs}
a_l^{(x,y)}(m+1) := a_l^{(x,y)}(m) + \dt \Big( \mu_l \lapl a_l^{(x,y)}(m) + \mathcal{L}^{(x,y)}_l\big(\vec{a}(m)\big) \Big), \qquad m \geq 0 .
\end{align}
Using the concept of ghost regions, this sequence corresponds to a finite difference scheme of the \emph{reaction-diffusion} PDE system (\ref{eq_pde}) that is subject to the NBCs
\begin{align}\label{eq_pde_nbcs}
0 & = \partial_x \alpha_l(1,y,t) &
0 & = \partial_x \alpha_l(0,y,t) &
0 & = \partial_y \alpha_l(x,1,t) &
0 & = \partial_y \alpha_l(x,0,t) ,
\end{align}
for all $1 \leq l \leq L$, $0 \leq t \leq T$, $x \in [0;1]$ and $y \in [0;1]$.

Unfortunately, the proof of Theorem~\ref{thm_scheme} does not carry over if reflective boundary conditions are assumed, meaning that we cannot establish the convergence of the finite difference scheme (\ref{eq_ode_euler_nbcs}). However, we are able to prove that (\ref{eq_ode_euler_nbcs}) is stable and consistent~\cite[Section 2.5 and 6.2]{thomas95} with respect to a certain norm. Informally, this means that the scheme arises from a Taylor approximation (consistency) and gives rise to an error that does not explode (stability).

\begin{theorem}\label{thm_scheme_nbcs}
Assume that for a fixed $T > 0$, the family of functions $(\alpha_1,\ldots,\alpha_L)$ on $[0;1]^2 \times [0;T]$ describes for all $0 < T' \leq T$ the unique solution of (\ref{eq_pde}) subjected to the NBCs (\ref{eq_pde_nbcs}) on $[0;1]^2 \times [0;T']$. Further, assume that
\[ \partial_{txx} \alpha_l, \ \partial_{xxxx} \alpha_l, \ \partial_{yxx} \alpha_l, \ \partial_{tyy} \alpha_l, \ \partial_{xyy} \alpha_l, \ \partial_{yyyy} \alpha_l, \]
where $1 \leq l \leq L$, exist and are continuous. For $M \geq 1$, define $\dt := T / M$, $r_l := \mu_l {\dt}/{\ds^2}$ and assume that $M$ is large enough such that $r_l \leq 1/4$ for all $1 \leq l \leq L$. Then, the Euler method (\ref{eq_ode_euler_nbcs}) interpreted as a finite difference scheme of (\ref{eq_pde}) is stable and consistent with respect to the ``averaging'' norm $\| \vec{a} \|_{1,\ds} :=  \sum_{(x,y) \in \mathcal{R}_K, 1 \leq l \leq L} | a_l^{(x,y)} | \ds^2$.
\end{theorem}

For linear PDEs, that is, if the reactive terms may be written in the form  $f_j(z_1, \ldots, z_L) = \sum_{l = 1}^L \xi_{j,l} z_l$, for $\xi_{j,l} \in \mathds{R}$, Theorem~\ref{thm_scheme_nbcs} and Lax's Theorem~\cite[Section 2.5 and 6.2]{thomas95} assert convergence.
\begin{corollary}\label{cor_scheme_nbcs}
If the PDE system (\ref{eq_pde}) is linear, then it is also well-posed~\cite[Section 6.2]{strikwerda89}. Let us further assume that the unique solution on $[0;1]^2 \times [0;T]$ is such that
\[ \partial_{txx} \alpha_l, \ \partial_{xxxx} \alpha_l, \ \partial_{yxx} \alpha_l, \ \partial_{tyy} \alpha_l, \ \partial_{xyy} \alpha_l, \ \partial_{yyyy} \alpha_l, \]
exist and are continuous for all $1 \leq l \leq L$.

Then, there exist $K_0, M_0 \geq 0$ and $C > 0$ such that the Euler method (\ref{eq_ode_euler_nbcs}) interpreted as a finite difference scheme of (\ref{eq_pde}) subjected to the NBCs (\ref{eq_pde_nbcs}) satisfies
\[ \max_{0 \leq m \leq M} \frac{1}{K^2} \sum_{\substack{(x,y) \in \mathcal{R}_K \\ 1 \leq l \leq L}} | \alpha_l(x,y,m\dt) - a^{(x,y)}_l(m) | \leq C(\ds + \dt)  \]
for all $K \geq K_0$ and $M \geq M_0$, where $\dt := T / M$.
\end{corollary}

Theorem \ref{thm_kurtz} and Corollary \ref{cor_scheme_nbcs} can now be used to show that the CTMC sequence converges in probability to the solution of the PDE system (\ref{eq_pde}) subjected to the NBCs (\ref{eq_pde_nbcs}) on $[0;1]^2 \times [0;T]$. Note that $\| \cdot \|_{1,\ds} \leq 4 L \| \cdot \|_{\infty}$, meaning that convergence in the supremum norm implies convergence in the averaging norm.

\begin{theorem}[Spatial PDE Limit]\label{thm_main_nbcs}
If the PDE system (\ref{eq_pde}) is linear, then it is also well-posed~\cite[Section 6.2]{strikwerda89} and the finite difference scheme (\ref{eq_ode_euler_nbcs}) satisfies
\[
\lim_{K \rightarrow \infty} \lim_{N \rightarrow \infty} \mathbb{P}  \Bigg\{ \sup_{0 \leq t \leq T} \frac{1}{K^2} \sum_{\substack{(x,y) \in \mathcal{R}_K \\ 1 \leq l \leq L}} \Big| \frac{1}{N}A_l^{(x,y)}(t) - \alpha_l(x,y,t) \Big|  > \varepsilon \Bigg\} = 0
\]
provided that $\partial_{txx} \alpha_l, \ \partial_{xxxx} \alpha_l, \ \partial_{yxx} \alpha_l, \ \partial_{tyy} \alpha_l, \ \partial_{xyy} \alpha_l, \ \partial_{yyyy} \alpha_l$, where $1 \leq l \leq L$, exist and are continuous on $[0;1]^2 \times [0;T]$.
%
\end{theorem}

\section{Case Study}\label{sec:validation}

The purpose of this section is to show how to use our PDE limit result for the performance modeling and analysis of a communication network with mobile nodes. In the following, we study in detail the mobile version of Example~\ref{ex:case-study} with its PDE limit (\ref{eq_pde_re}) using absorbing boundary conditions.

\paragraph*{System Description}
Our mobile network consists of a lattice of $(K+1)^2$ regions where each region represents an area offering Internet connectivity to all nodes therein located by means of an 802.11 access point. We assume that nodes feature an \emph{on}/\emph{off} behavior whereby they interpose an exponentially distributed think time between successive connections. Whilst pausing, nodes may decide to move across one of the neighboring regions according to an unbiased RW. When connected, a node downloads a file with an exponentially distributed size, and does not move until the download has finished. We assume that the overall network's bottleneck is represented by the link between the node and the access point, therefore the performance depends on the current load at each region and on how the capacity is shared amongst the nodes in a region.

This system was implemented in the JiST framework, using the SWANS module for the modeling of wireless and ad-hoc networks~\cite{Barr:2005:JEA:1060168.1060170}. The physical layer was configured using default values for all parameters. For simplicity, we used a model of transmission with no packet loss. Because of the scenario considered in this study, nodes were not equipped with a dynamically updating routing protocol because there was no direct communication between nodes, and the destination of a packet could be determined based on the node location, by maintaining a simple mapping between the current location and the MAC address of the access point. The application layer was implemented as a lightweight file transfer protocol on top of UDP\@. The file size to download was exponentially distributed with mean 40\,KB, and the file was transmitted in 2\,KB packets. We assume that the think time and the RW behavior during the \emph{off} state of the nodes are given by independent distributions. Let $\lambda$ be the think rate and $\mu K^2 $ be the movement rate to each neighboring region. Thus, the holding time in the \emph{off} state is exponentially distributed with mean $1/(\lambda+ 4 \mu K^2 )$.

\begin{table*}[t]
\centering
\begin{tabular}{rrrrrrrr}
\toprule
 & & \multicolumn{3}{c}{$\mu = 0.001$} & \multicolumn{3}{c}{$\mu = 0.010$} \\
\cmidrule(l){3-5} \cmidrule(l){6-8}
\multicolumn{1}{c}{$V$}  & \multicolumn{1}{c}{\emph{Nodes}} & \multicolumn{1}{c}{\emph{JiST}} & \multicolumn{1}{c}{\emph{CTMC \ (Error)}} & \multicolumn{1}{c}{\emph{PDE \ (Error)}} &   \multicolumn{1}{c}{\emph{JiST}} & \multicolumn{1}{c}{\emph{CTMC \ (Error)}} & \multicolumn{1}{c}{\emph{PDE \ (Error)}} \\
\midrule
10  & 56 & 0.8852 & 0.8638  (0.0214) & 0.9007 (0.0155) & 0.5161 & 0.5368 (0.0207) & 0.2716  (0.2445) \\
25 & 156 & 0.6968 & 0.6939 (0.0029) & 0.7186 (0.0218) & 0.5207 & 0.5480  (0.0273) & 0.3807 (0.1400) \\
50  &320  & 0.4734 & 0.4578 (0.0156) & 0.4646 (0.0088) & 0.3926 & 0.3742 (0.0184) & 0.2840 (0.1086) \\
\bottomrule
\end{tabular}
\caption{Validation results: Average fraction of nodes in the \emph{off} state after 10 time units in an area divided in 64 regions ($K = 7$), for $N = 1$. The errors are given as the absolute difference with respect to the JiST estimated average.}\label{tab:validation}
\end{table*}

\begin{table*}[t]
\centering
\begin{tabular}{rrrccrccrcc}
\toprule
& &   \multicolumn{3}{c}{$V = 10$} & \multicolumn{3}{c}{$V = 25$} & \multicolumn{3}{c}{$V = 50$} \\
\cmidrule(l){3-5} \cmidrule(l){6-8} \cmidrule(l){9-11}
$N$ & $K$ & \multicolumn{1}{c}{\emph{Nodes}}  & \multicolumn{1}{c}{\emph{CTMC}} &
\multicolumn{1}{c}{\emph{Err. PDE}} & \multicolumn{1}{c}{\emph{Nodes}}  & \multicolumn{1}{c}{\emph{CTMC}} &
\multicolumn{1}{c}{\emph{Err. PDE}}  & \multicolumn{1}{c}{\emph{Nodes}} &
\multicolumn{1}{c}{\emph{CTMC}} & \emph{Err. PDE} \\
\midrule
1 & 7  & 56 & 0.5368   & 0.2652 & 156 & 0.5480 &  0.1673  & 320 & 0.3742 & 0.0902 \\
2 & 7 & 124 & 0.5073 &  	0.2350 &  320 & 0.5500 & 0.1693 & 660 & 0.4172 & 0.1332 \\
\midrule
1 & 15 & 268  & 0.3897 & 0.1181 & 716 & 0.4611 & 0.0804 & 1496 & 0.3397 &  0.0557  \\
2 & 15 & 546  & 0.3897 & 0.1181 & 1496	& 0.4658 & 0.0851 & 3028 & 0.3415 &  0.0575 \\
\midrule
1 & 31 & 1132 & 0.3424  & 0.0708 & 3100 & 0.4139 &0.0332   & 6338 & 0.3102 &  0.0262  \\
2 & 31 & 2248 & 0.3341  & 0.0625 & 6388 &  0.4256 & 0.0449  & 12948 &	0.3061 &  0.0221 \\
\midrule
1 & 47 & 2592 &  0.3238 & 0.0522 &   7108 & 0.3984 &0.0177 & 14628 & 0.3035 & 0.0195 \\
2 & 47 &	5572	& 0.3096 & 0.0380 & 14628 &  0.4046 &0.0239 & 29736& 0.3011 & 0.0171  \\
\midrule
1 &  63 & 4676 &  0.3049  & 0.0335 & 12736 & 0.3893 & 0.0009 & 26296& 0.2928 & 0.0088 \\
2 &  63 & 10084   &  0.3016   & 0.0300  & 29296 &  0.3943 & 0.0136 & 53480 & 0.2973 & 0.0133 \\
\midrule
1 & 79 & 7268 & 0.3072 & 0.0356 & 19996  & 0.3813 & 0.0006 & 41392  & 0.2907 & 0.0067 \\
2 & 79 & 	15736 & 0.2943 & 0.0227 & 41312 	    & 0.3877 & 0.0070 & 84084 & 0.2910 & 0.0070  \\
\bottomrule
\end{tabular}
\caption{Scaling behavior for $\mu = 0.010$. The PDE approximation error is the absolute difference with respect to the CTMC estimate.}\label{tab:scaling}
\end{table*}

\paragraph*{Stochastic Model} To model this system, we use the mobile reaction network with interaction functions given in Example~\ref{ex:case-study}, where we let $\mu_1^K = \mu K^2$, for some $\mu > 0$ which will be varied in our numerical experiments, and $\mu_2^K = 0$. The scaling of $\mu^K_1$ satisfies assumption \textbf{(A3)}, while the choice of $\mu_2^K$ models the fact that the nodes stand still until the download has finished. With this, we wish to highlight that the unbiased RW can indeed be made dependent upon the local state of the node; analogous results to those presented in this section would be obtained with $\mu_2^K > 0$, i.e., allowing the node to transfer to another region even during the download process. By varying $\mu$ in our experiments, instead, we wish to show the impact of the mobility model on the network performance.

\paragraph*{Remarks} Let us observe that, in essence, our stochastic model is a \emph{spatial} version of a queueing network with exponentially distributed service times. In each region $(x,y)$, $A_1^{(x,y)}$ represents the queue length at a delay station whereas $A_2^{(x,y)}$ is the queue length at a multi-server station with $N$ servers with individual rate $c$. A job serviced at this station goes to the delay station in the same region with probability 1, whereas a job at the delay station may go into service in the same region, or into a neighboring delay station according to the movement rate $\mu^1_K$. The knowledge of the queue lengths in each region, given by the PDE solution, therefore provides a complete characterization of the system's performance.

Due to the assumption on the boundary conditions and to the forms of $F_1$ and $F_2$, all nodes will eventually leave the spatial domain with probability 1. Although we make use of this scenario in the present paper, we wish to remark that our framework also supports exogenous arrivals. For this, however, the arrival rate must scale linearly with $N$. For instance, interaction 1 in Example~\ref{ex:bit-torrent}  describes the arrival of $A_1$-type nodes into each region with rate $N \lambda$. This models Poisson arrivals with intensity proportional to some parameter function $\lambda(x,y)$ that must be consistent with the DBCs.

\paragraph*{Model Validation} We tested the validity of our stochastic model against the reference  behavior as given by discrete-event simulation of the JiST implementation. The performance metric of interest was chosen to be the ratio between the expected total number of nodes in state $A_1$ and the initial population of nodes after 10 time units. The CTMC was solved by simulation, since closed-form expressions for the transient behavior are not available and numerical CTMC solution was made infeasible by the excessively large state spaces. In both cases, simulations were conducted using the method of independent replicas; the stopping criterion was the convergence of the 95\% confidence interval within 5\% of the mean. Throughout the remainder of this section, we report such means. Finally, the PDEs were solved using Matlab and its \emph{Partial Differential Equation Toolbox} with the function \texttt{parabolic}, which implements the numerical solution of systems of reaction-diffusion PDE with both NBCs and DBCs.

The JiST implementation and the CTMC model were parametrized as follows. Throughout all tests, we kept $\lambda$ fixed at 0.250. We estimated the parameter $c = 2.857$ by measuring the download times in a network with a single node, that is, in the absence of contention. The initial population of nodes across the spatial domain was generated according to the function $\theta$ given in (\ref{eq_theta}). The initial concentrations of nodes were set as
\[
\alpha_1^0(x,y) = 0 \text{\qquad and \qquad} \alpha_2^0 = V \theta(x,y),
\]
where $V$ is a parameter that was varied in our tests to observe the network under different node densities. Thus, the $N$-th element of the CTMC sequence had an initial number of downloading nodes in region $(x,y)$ equal to $\lfloor N V \theta(x,y) \rfloor$.

\paragraph*{Results} Table~\ref{tab:validation} shows the results of the comparison between the average fraction of nodes in the \emph{off} state after 10 time units as computed by the JiST discrete-event simulation, CTMC simulation, and PDE analysis, for different values of $V$ and $\mu$, and for $N = 1$ and $K = 7$ (corresponding to an area with 64 regions). As can be expected, our parametrizations exercise the network under a range of different operating regimes. This can be noticed by the fact that our performance index varies sensibly with $V$ in Table~\ref{tab:validation}.  For any fixed $V$, the estimates for the different values of $\mu$ numerically show the impact of the mobility model on the network performance. We remark that the observation window of interest, 10 time units, is in all cases sufficient to study the network under interesting conditions where the fraction of nodes in the \emph{off} state is significantly different than~0 (which corresponds to the initial condition).

The table shows good agreement, in all cases, between the JiST estimates and the CTMC model, according to the notion of error defined as the absolute difference with respect to the JiST average.  The PDE approximation, instead, yielded different quality depending on the choice of $\mu$. Indeed, it already provides good accuracy for $\mu = 0.001$, whereas it suffered more sever errors for $\mu = 0.010$. This is, however, unsurprising because the PDE limit holds for large $N$ and $K$.
To numerically show convergence to the PDE limit, we considered all cases of $V$ with  $\mu = 0.010$ and applied the space scaling by simulating the CTMCs for different values of $K$ until 79, thus dividing the space in at most 6400 regions. The numerical results are presented in Table~\ref{tab:scaling}, where the PDE approximation error is measured as the absolute difference with respect to the corresponding average estimate computed by CTMC simulation. It is interesting to note that, for a fixed $K$, the PDE approximation error may tend to increase with $N$, see, for instance, $K=31$ for $V=25$. However, this is not in contradiction with our limit result which holds when both $N$ and $K$ are large. Indeed, in all cases, increasing $K$ leads to sensibly increased accuracy, with errors that are negligible for all practical purposes for $K=79$, corresponding to a situation with at most ca 84,000 nodes in the network. There, we also verified that the means for $N=1$ and $N=2$ were statistically equal at the 0.05 level.

\begin{figure}[t]
\centering
\includegraphics[scale=0.75]{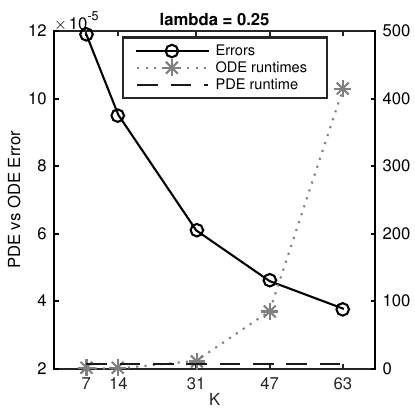}
\hfill
\includegraphics[scale=0.75]{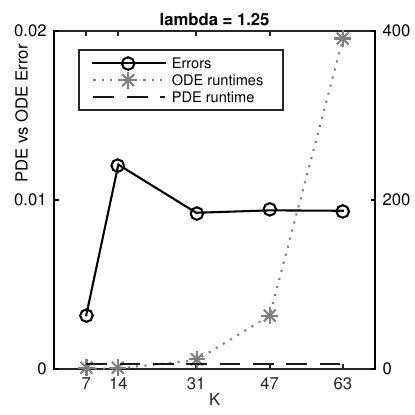}
\hfill
\includegraphics[scale=0.75]{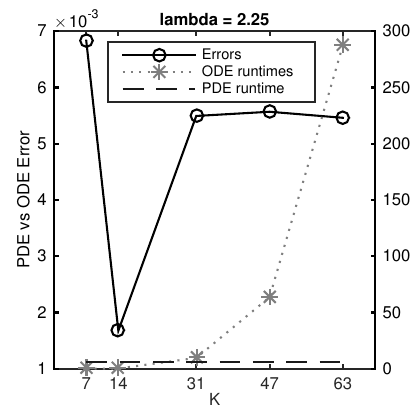}
\caption{Comparison of ODE/PDE estimation errors and runtimes for different values of $K$ and $\lambda$ (0.25, left plot, 1.25 middle plot, and 2.25, right plot; the left $y$ axis shows the absolute error between the ODE solution and the PDE solution of the fraction of agents in the \emph{off} state after 10.0 time units; the right $y$ axis shows the runtimes (in seconds) of the ODE solution with varying $K$ relative to the PDE runtime (dashed line).}\label{fig:odevspde}
\end{figure}
\paragraph*{ODE/PDE Comparison} As discussed, the main advantage in using the limit PDE instead of the limit ODE system is that the mesh discretization becomes a parameter of the solver, instead of being a parameter of the model (i.e., the actual number of regions in the spatial domain). An ODE/PDE comparison is proposed in Figure~\ref{fig:odevspde}, where we consider the absolute error between the the ODE solution and the PDE solution of the fraction of agents in the \emph{off} state after 10.0 time units, together with their runtimes. This was done by fixing $V = 10$ and varying $\lambda = 0.25, 1.25, 2.25$.  Using Matlab's PDE solver default settings, in all cases considered in this paper the PDEs were analyzed within 8\,s on an ordinary machine (consistently with~\cite{valuetools14mtmt}, where a similar analysis was conducted on epidemiological models). This turned out to be significantly cheaper from a computational point of view than directly solving the system of $2(K+1)^2$ ODEs  (with Matlab's \texttt{ode15s} function). In particular,  for $K = 63$ we registered ODE runtimes two orders of magnitude larger than the PDE solution times, whilst yielding accurate results in all cases (up to at most ca 1\%). Clearly, instead, CTMC simulations required a sheer amount of computational power --- for instance, we did measure runtimes of ca 32 hours for the CTMC simulation for the cases $N=2$ and $K=79$, with independent replicas being run in parallel on eight cores.

\section{Conclusion}\label{sec:conclusion}

In this paper we have proposed an approach to analyze stochastic models of mobile networks by means of a deterministic approximation represented by a system of partial differential equations. These offer a macroscopic, continuous view of the network dynamics which holds in the limit of infinite populations of nodes in the network and infinite number of regions in a regular two-dimensional spatial lattice. In practice, in our numerical tests we observed good accuracy in many cases even with few nodes and coarse lattices, by carrying out a validation study against Markov-chain simulations and detailed discrete-event simulations at the network-protocol level. Furthermore, we found good speed of convergence with increasing system sizes, yielding high accuracy for networks consisting of thousands of nodes. These observations allow us to conclude that partial differential equations can be effectively used for the analysis of large-scale mobile systems.

\section*{Acknowledgement}
This work was partially supported by the EU project QUANTICOL, 600708.

\bibliographystyle{IEEEtran}
\bibliography{max,mirco}

\appendix



\section{Proofs}

\begin{proof}[Proof of Theorem \ref{thm_kurtz}]
Let us denote by $E_N$ the state space of $\big(\frac{1}{N } \vec{A}_N(t)\big)_{t\geq0}$ subjected to $\frac{1}{N } A^{(x,y)}_l(0) = \lfloor N \alpha_l^0(x,y) \rfloor / N$. For an $\underline{\vec{a}} \in E_N$, let $\lambda_N(\underline{\vec{a}})$ denote the waiting time distribution in state $\underline{\vec{a}}$, meaning that
\[
\mathbb{P} \Big\{ \tau({\underline{\vec{a}}}) > t \Big| \frac{1}{N } \vec{A}_N(0) = \underline{\vec{a}} \Big\} = e^{- \lambda_N(\underline{\vec{a}})t }, \quad \text{where} \quad
\tau({\underline{\vec{a}}}) := \inf\Big\{ t \geq 0 \Big| \frac{1}{N } \vec{A}_N(t) \neq \underline{\vec{a}} \Big\} .
\]
For every set $\Gamma$ in the Borel $\sigma$-algebra of $E_N$ we  further define the jump distribution in $\underline{\vec{a}}$ as
\[ \mu_N(\underline{\vec{a}}, \Gamma) := \mathbb{P} \Big\{ \frac{1}{N } \vec{A}_N(\tau({\underline{\vec{a}}})) \in \Gamma \Big| \frac{1}{N } \vec{A}_N(0) = \underline{\vec{a}} \Big\} . \]
Then, the drift $D_N(\underline{\vec{a}})$ of $\big(\frac{1}{N } \vec{A}_N(t)\big)_{t\geq0}$ in the state $\underline{\vec{a}} \in E_N$ can be written in the form
\begin{equation*}
\begin{split}
D_N (\underline{\vec{a}}) & = \lambda_N(\underline{\vec{a}}) \int_{\underline{\vec{b}} \in E_N} (\underline{\vec{b}} - \underline{\vec{a}}) \mu_N(\underline{\vec{a}}, d\underline{\vec{b}})
= \int_{\underline{\vec{b}} \in E_N} (\underline{\vec{b}} - \underline{\vec{a}}) \lambda_N(\underline{\vec{a}}) \mu_N(\underline{\vec{a}}, d\underline{\vec{b}}) \\
& = \sum_{(x,y)} \bigg[ \sum_{l,(\tilde{x},\tilde{y})} \frac{w_l^{(x,y),(\tilde{x},\tilde{y})}}{N }\mu_l^K \big(N  \underline{a}_l^{(x,y)}\big) + \sum_j \frac{w_j^{(x,y)}}{N } \cdot
F_j\big(N \underline{a}_1^{(x,y)}, \ldots, N \underline{a}_L^{(x,y)}, \beta_1(x,y), \ldots, \beta_P(x,y) \big)  \bigg] ,
\end{split}
\end{equation*}
where $w_j^{(x,y)}$ and $w_l^{(x,y),(\tilde{x},\tilde{y})}$ denote the corresponding jump vectors of $(\vec{A}_N(t))_{t\geq0}$. For instance, $w_l^{(\ds,\ds),(2\ds,\ds)}$ decrements $A^{(\ds,\ds)}_l$ and increments $A^{(2\ds,\ds)}_l$, whereas $w_l^{(\ds,\ds),(0,\ds)}$ just decrements $A^{(\ds,\ds)}_l$. Also, $w_j^{(x,y)} = 0 = w_l^{(x,y),(\tilde{x},\tilde{y})}$ whenever $(x,y) \in \Omega$. Note that
\[
\Bigg( \sum_{(x,y)} \bigg[ \sum_{l,(\tilde{x},\tilde{y})} w_l^{(x,y),(\tilde{x},\tilde{y})} \mu_l^K \underline{a}_l^{(x,y)} + \sum_j w_j^{(x,y)} f_j\big(\underline{a}_1^{(x,y)}, \ldots, \underline{a}_L^{(x,y)}, \beta_1(x,y), \ldots, \beta_P(x,y)\big) \bigg] \Bigg)^{(x,y)}_l = \mathds{1}(x,y) \cdot \Big( \mathcal{M}^{(x,y)}_l(\underline{\vec{a}}) + \mathcal{L}^{(x,y)}_l(\underline{\vec{a}}) \Big)
\]
for all $\underline{\vec{a}} \in E_N$, $1 \leq l \leq L$ and $(x,y) \in \mathcal{R}$. Moreover, since the theorem of Picard-Lindel\"{o}f and \textbf{(A2)} assert that (\ref{eq_ode_sys}) has a unique solution $\vec{a}$ in $\mathbb{R}^{\mathcal{R} \times \{1, \ldots, L\}}$ and the time domain of $\vec{a}$ contains $[0;T]$ by assumption, it holds that $c := \sup_{0 \leq t \leq T} \| \vec{a}(t) \| + 1 < \infty$. Thanks to the compactness of $E := [-c;c]^{\mathcal{R} \times \{1, \ldots, L\}}$ and $\textbf{(A1)}$, the above implies
\begin{equation*}
\lim_{N \to \infty} \sup_{\underline{\vec{a}} \in \Gamma_N}
\bigg\| D_N(\underline{\vec{a}}) - \sum_{(x,y)} \bigg[ \sum_{l,(\tilde{x},\tilde{y})} w_l^{(x,y),(\tilde{x},\tilde{y})} \mu_l^K \underline{a}_l^{(x,y)} + \sum_j w_j^{(x,y)} f_j\big(\underline{a}_1^{(x,y)}, \ldots, \underline{a}_L^{(x,y)}, \beta_1(x,y), \ldots, \beta_P(x,y)\big) \bigg] \bigg\| = 0 ,
\end{equation*}
where $\Gamma_N := E_N \cap E$. Moreover, it holds that
\[
E_N \cap \Big\{ \vec{b} \in \mathbb{R}^{\mathcal{R} \times \{1, \ldots, L\}} \ \big\vert \ \inf_{0 \leq t \leq T} \| \vec{b} - \vec{a}(t) \| \leq \frac{1}{2} \Big\} \subseteq E
\qquad \text{and} \qquad
\sup_{N \geq 1} \sup_{ \underline{\vec{a}} \in \Gamma_N } \lambda_N(\underline{\vec{a}}) \int_{\underline{\vec{b}} \in E_N} \| \underline{\vec{b}} - \underline{\vec{a}} \| \mu_N(\underline{\vec{a}}, d\underline{\vec{b}}) < \infty .
\]
To see this, we note that
\begin{align*}
\lambda_N(\underline{\vec{a}}) \int_{\underline{\vec{b}} \in E_N}  \| \underline{\vec{b}} - \underline{\vec{a}} \|  \mu_N(\underline{\vec{a}}, d\underline{\vec{b}}) & \leq \sqrt{(K+1)^2L} \int_{\underline{\vec{b}} \in E_N} \| \underline{b} - \underline{a} \|_\infty \lambda_N(\underline{\vec{a}}) \mu_N(\underline{\vec{a}}, d\underline{\vec{b}}) \\
& \leq (K+1)\sqrt{L} \sum_{(x,y)} \bigg[ \sum_j \frac{C}{N} F_j\big(N \underline{a}_1^{(x,y)}, \ldots, N \underline{a}_L^{(x,y)}, \vec{\beta}(x,y)\big)
+ \sum_l \frac{1}{N} 4 \mu_l^K N \underline{a}_l^{(x,y)} \bigg] ,
\end{align*}
where $C := \max\big\{ | d_{jl} - c_{jl} | \mid 1 \leq j \leq J \land 1 \leq l \leq L \big\}$. Moreover, for all $N \geq 1$, it holds that
\[
\sup_{ \underline{\vec{a}} \in \Gamma_N } \lambda_N(\underline{\vec{a}}) \int_{ \| \underline{\vec{b}} - \underline{\vec{a}} \| > \varepsilon_N } \| \underline{\vec{b}} - \underline{\vec{a}} \| \mu_N(\underline{\vec{a}}, d\underline{\vec{b}}) = 0, \qquad \text{where} \quad \varepsilon_N := \frac{(C+1)(K+1)\sqrt{L}}{N},
\]
because $\varepsilon_N < \| \underline{\vec{b}} - \underline{\vec{a}} \| \leq \sqrt{(K+1)^2 L} \| \underline{\vec{b}} - \underline{\vec{a}} \|_\infty$ implies
$\frac{C+1}{N} < \| \underline{\vec{b}} - \underline{\vec{a}} \|_\infty $. This allows us to conclude the claim using Theorem 2.11 in \cite{kurtz-1970}.
\end{proof}


\begin{proof}[Proof of Theorem \ref{thm_kurtz_nbcs}]
Similarly to the proof of Theorem~\ref{thm_kurtz}.
\end{proof}


\begin{proof}[Proof of Proposition \ref{prop_scaling}]
Using the concept of characteristic function, it is possible to establish an analytic expression of the Fourier-Laplace transform of the transient probability distribution of a continuous-time random walk (CTRW) in the case of general holding times. This, in turn, characterizes $\mathbb{E}(d_k^2(t))$ in terms of the underlying Laplace transform, see \cite[Section 2.5]{WeissRandomWalk}. However, in the case of exponentially distributed holding times, it is possible to provide a direct proof that relies on the uniformization method for CTMCs with countable state spaces \cite[Chapter 8, Section 2.1]{bremaud}. For this, let the CTRW $(W_k(t))_{t\geq0}$ be given by means of the standard graphical notation in Figure \ref{fig_CTRW} and let $(\widehat{W}_k(n))_{n\geq0}$ denote the unbiased random walk in \emph{discrete} time on $\frac{1}{k} \mathbb{Z} \times \frac{1}{k} \mathbb{Z}$ with $\widehat{W}_k(0) = (0,0)$. Further, let $(N(t))_{t \geq 0}$ be the homogenous Poisson process with intensity $r_k$ and $P$ denote the transition matrix of $(\widehat{W}_k(n))_{n\geq0}$. Then it holds that $W_k(t) = \widehat{W}_k(N(t))$ and
\begin{equation*}
\begin{split}
\mathbb{E}(d_k^2(t)) & = \sum_{(x,y) \in \mathbb{Z}^2} \mathbb{P}\Big\{W_k(t) = \Big(\frac{x}{k},\frac{y}{k}\Big)\Big\} \Big[ \Big(\frac{x}{k}\Big)^2 + \Big(\frac{y}{k}\Big)^2 \Big]
= \frac{1}{k^2} \sum_{(x,y) \in \mathbb{Z}^2} \bigg( \sum_{n=0}^{\infty} \vec{e}_{\left(\frac{0}{k},\frac{0}{k}\right)}^{\ T} P^n \frac{(r_k t)^n}{n!} e^{-r_k t} \bigg) \vec{e}_{\left(\frac{x}{k},\frac{y}{k}\right)} \big( x^2 + y^2 \big) \\
& = \frac{1}{k^2} e^{-r_k t} \sum_{n=0}^{\infty} \frac{(r_k t)^n}{n!} \bigg( \sum_{(x,y) \in \mathbb{Z}^2} \vec{e}_{\left(\frac{0}{k},\frac{0}{k}\right)}^{\ T} P^n \vec{e}_{\left(\frac{x}{k},\frac{y}{k}\right)} \big( x^2 + y^2 \big) \bigg)
= \frac{1}{k^2} e^{-r_k t} \sum_{n=0}^{\infty} \frac{(r_k t)^n}{n!} \mathbb{E}\big( \| \widehat{W}_k(n) \|^2\big) \\
& = \frac{1}{k^2} e^{-r_k t} \sum_{n=0}^{\infty} \frac{(r_k t)^n}{n!} n = \frac{1}{k^2} e^{-r_k t} r_k t \sum_{n=1}^{\infty} \frac{(r_k t)^{n-1}}{(n-1)!} = \left(\frac{r_k}{k^2}\right) t .
\end{split}
\end{equation*}
\end{proof}

\begin{figure}
\centering
    \begin{tikzpicture}[->,>=stealth',shorten >=1pt,auto,semithick]
    \matrix[row sep=0.8cm, column sep=0.6cm]{

    \node (1) {$$}; \pgfmatrixnextcell \node (2) {$\ldots$}; \pgfmatrixnextcell \node (3) {$\ldots$}; \pgfmatrixnextcell \node (4) {$\ldots$}; \pgfmatrixnextcell \node (5) {$$}; \\
    \node (6) {$\ldots$}; \pgfmatrixnextcell \node (7) {$(\frac{-1}{k},\frac{1}{k})$}; \pgfmatrixnextcell \node (8) {$(\frac{0}{k},\frac{1}{k})$}; \pgfmatrixnextcell \node (9) {$(\frac{1}{k},\frac{1}{k})$}; \pgfmatrixnextcell \node (10) {$\ldots$}; \\
    \node (11) {$\ldots$}; \pgfmatrixnextcell \node (12) {$(\frac{-1}{k},\frac{0}{k})$}; \pgfmatrixnextcell \node (13) {$(\frac{0}{k},\frac{0}{k})$}; \pgfmatrixnextcell \node (14) {$(\frac{1}{k},\frac{0}{k})$}; \pgfmatrixnextcell \node (15) {$\ldots$}; \\
    \node (16) {$\ldots$}; \pgfmatrixnextcell \node (17) {$(\frac{-1}{k},\frac{-1}{k})$}; \pgfmatrixnextcell \node (18) {$(\frac{0}{k},\frac{-1}{k})$}; \pgfmatrixnextcell \node (19) {$(\frac{1}{k},\frac{-1}{k})$}; \pgfmatrixnextcell \node (20) {$\ldots$}; \\
    \node (21) {$$}; \pgfmatrixnextcell \node (22) {$\ldots$}; \pgfmatrixnextcell \node (23) {$\ldots$}; \pgfmatrixnextcell \node (24) {$\ldots$}; \pgfmatrixnextcell \node (25) {$$}; \\
    };

    \path
            (6) [bend left] edge node {$$} (7)
            (7) edge node {$$} (6)
            (7) edge node {$$} (8)
            (8) edge node {$$} (7)
            (8) edge node {$$} (9)
            (9) edge node {$$} (8)
            (9) edge node {$$} (10)
            (10) edge node {$$} (9)
            (11) edge node {$$} (12)
            (12) edge node {$$} (11)
            (12) edge node {$$} (13)
            (13) edge node {$$} (12)
            (13) edge node {$$} (14)
            (14) edge node {$$} (13)
            (14) edge node {$$} (15)
            (15) edge node {$$} (14)
            (16) edge node {$$} (17)
            (17) edge node {$$} (16)
            (17) edge node {$$} (18)
            (18) edge node {$$} (17)
            (18) edge node {$$} (19)
            (19) edge node {$$} (18)
            (19) edge node {$$} (20)
            (20) edge node {$$} (19)
            (2) [bend left] edge node {$$} (7)
            (7) edge node {$$} (2)
            (7) edge node {$$} (12)
            (12) edge node {$$} (7)
            (12) edge node {$$} (17)
            (17) edge node {$$} (12)
            (17) edge node {$$} (22)
            (22) edge node {$$} (17)
            (3) edge node {$$} (8)
            (8) edge node {$$} (3)
            (8) edge node {$$} (13)
            (13) edge node {$$} (8)
            (13) edge node {$$} (18)
            (18) edge node {$$} (13)
            (18) edge node {$$} (23)
            (23) edge node {$$} (18)
            (4) edge node {$$} (9)
            (9) edge node {$$} (4)
            (9) edge node {$$} (14)
            (14) edge node {$$} (9)
            (14) edge node {$$} (19)
            (19) edge node {$$} (14)
            (19) edge node {$$} (24)
            (24) edge node {$$} (19);
    \end{tikzpicture}
\caption{Continuous-time random walk $(W_k(t))_{t\geq0}$ on the lattice $\frac{1}{k} \mathbb{Z} \times \frac{1}{k} \mathbb{Z}$. The rate of each transition is $r_k/4$ and was suppressed to enhance readability. In particular, the average sojourn time in each state is $1/r_k$.}\label{fig_CTRW}
\end{figure}
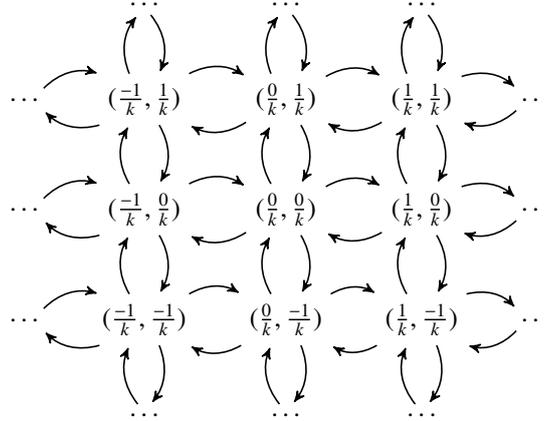

\begin{proposition}\label{prop_loc_glob_lip}
Let us fix an $f : \mathbb{R}^n \rightarrow \mathbb{R}$ which is locally Lipschitz and some $c > 0$. Then, the function
\begin{equation*}
\begin{split}
& \underline{f}(x) :=
\begin{cases}
f(x) & \ , \| x \| \leq c \\
f\big(c \frac{x}{\| x \|}\big) & \ , \text{otherwise}
\end{cases}
\end{split}
\end{equation*}
is globally Lipschitz on $\mathbb{R}^n$.
\end{proposition}

\begin{proof}
As the function $f$ is locally Lipschitz on $\mathbb{R}^n$ and the set $\{ z \in \mathbb{R}^n \mid \|z\| \leq c \}$ is compact, there exists a $\Lambda > 0$ such that
\[ | f(y) - f(x) | \leq \Lambda \| y - x \| \]
for all $x, y \in \{ z \in \mathbb{R}^n \mid \|z\| \leq c \}$. In the following, we show that
\[ | \underline{f}(y) - \underline{f}(x) | \leq \Lambda \| y - x \| \]
for all $x, y \in \mathbb{R}^n$. Let us fix for this two arbitrary $x,y \in \mathbb{R}^n$.

We note that the cases $(x = 0 \land y = 0)$, $(x = 0 \land y \neq 0)$ and $(x \neq 0 \land y = 0)$ are obvious. For the case $x \neq 0 \land y \neq 0$, we may assume without loss of generality that $\|x\| \leq \|y\|$. Then, the case can be divided into the subcases $(\|x\| < c \land \|y\| < c),$ $(\|x\| \geq c \land \|y\| \geq c)$ and $(0 < \|x\| < c \leq \|y\|)$. The first subcase is clear, as $\underline{f} \equiv f$ on $\{ x \in \mathbb{R}^n \mid \| x \| \leq c \}$.

Let us consider now the case $(\|x\| \geq c \land \|y\| \geq c)$. Specifically, we assume first that $y$ and $x$ are linearly dependent, i.e. $y = \lambda x$ for some $\lambda \in \mathbb{R}$. As $\underline{f}(y) = \underline{f}(x)$ in the case of $\lambda \geq 0$, we may assume that $\lambda < 0$. Then, as
\begin{equation*}
| \underline{f}(y) - \underline{f}(x) | = \left| f\left( \frac{c \lambda x}{|\lambda| \|x\|} \right) - f\left(\frac{cx}{\|x\|}\right) \right|
\leq \Lambda \left\| - \frac{cx}{\|x\|} - \frac{cx}{\|x\|} \right\| = \Lambda \frac{c}{\|x\|} 2 \|x\| = \Lambda 2 c,
\end{equation*}
the estimation $\|y - x\| = (1 + |\lambda|) \| x \| = \|x\| + \|y\| \geq 2c$ yields the claim.

Let us assume now that $x$ and $y$ are linearly independent. Then, the calculation
\begin{equation*}
| \underline{f}(y) - \underline{f}(x) | = \left| f\left(\frac{cy}{\|y\|}\right) - f\left(\frac{cx}{\|x\|}\right) \right| \leq \Lambda \left\| \frac{cy}{\|y\|} - \frac{cx}{\|x\|} \right\|
= \Lambda \frac{c}{\|x\|} \left\| \frac{\|x\|}{\|y\|} y - x \right\| \leq \Lambda \left\| \frac{\|x\|}{\|y\|} y - x \right\|
\end{equation*}
shows that it is sufficient to prove that $\| p - x \| \leq \| y - x \|$, where $p := \frac{\|x\|}{\|y\|} y$. An informal pictorial description of the situation in the two-dimensional space is given in Figure \ref{prop_loc_glob_lip_fig}. The calculation

\begin{figure}
\begin{center}
    \includegraphics[scale=0.9]{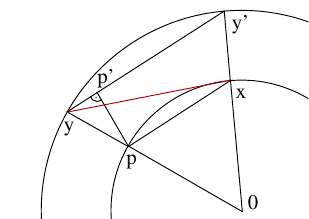}\hfill
    \caption{Graphical illustration in the two-dimensional case for the proof idea of Proposition~\ref{prop_loc_glob_lip} to show the  Lipschitz continuity of $\underline{f}$ in the case of linearly independent vectors $x$ and $y$.}\label{prop_loc_glob_lip_fig}
\end{center}
\end{figure}

\begin{multline*}
\| y - x \|^2 = \| y - p + p - x \|^2 = \langle y - p + p - x,  y - p + p - x \rangle =
\\ = \langle y - p, y - p \rangle + 2 \langle y - p, p - x \rangle + \langle p - x, p - x\rangle = \|y - p\|^2 + \| x - p \|^2 - 2 \langle y - p, x - p \rangle
\end{multline*}
shows then that $\| x - p \| \leq \| y - x \| $ if $\langle y - p, x - p \rangle \leq 0$. (Note that the above calculation establishes nothing more than the general version of the cosine law.) In order to prove this, we introduce the auxiliary points $y' := \frac{\|y\|}{\|x\|} x$ and
\[ p' := y + \left\langle \frac{y' - y}{\|y' - y\|}, p - y \right\rangle \frac{y' - y}{\|y' - y\|} . \]
Note that the linear independence of $x$ and $y$ implies $y' - y \neq 0$. Hence, $p'$ is well-defined and we conclude
\begin{equation*}
\langle y - p, x - p \rangle = \langle y - p' + p' - p, x - p \rangle = \langle y - p', x - p \rangle + \langle p' - p, x - p \rangle.
\end{equation*}
Consequently, it is sufficient to show that $\langle y - p', x - p \rangle \leq 0$ and $\langle p' - p, x - p \rangle = 0$. The latter follows from
\begin{equation*}
\begin{split}
\frac{\|y\|}{\|x\|} \left\langle p' - p, x - p \right\rangle & = \left\langle p' - p, \frac{\|y\|}{\|x\|}(x - p) \right\rangle = \langle p' - p, y' - y \rangle
= \left\langle y + \left\langle \frac{y' - y}{\|y' - y\|}, p - y \right\rangle \frac{y' - y}{\|y' - y\|} - p, y' - y \right\rangle \\
& = \langle y - p, y' - y\rangle + \left\langle \frac{y' - y}{\|y' - y\|} , p - y \right\rangle \left\langle \frac{y' - y}{\|y' - y\|} , y' - y \right\rangle \\
& = \langle y - p, y' - y\rangle + \langle y' - y , p - y \rangle \frac{\langle y' - y , y' - y \rangle}{\|y' - y\|^2} \\
& = \langle y - p, y' - y\rangle - \langle y - p , y' - y \rangle = 0 ,
\end{split}
\end{equation*}
whereas the calculation
\begin{equation*}
\begin{split}
\langle y - p', x - p \rangle & = \left \langle - \left\langle \frac{y' - y}{\|y' - y\|}, p - y \right\rangle \frac{y' - y}{\|y' - y\|} , x - p \right\rangle
= \left\langle \frac{y' - y}{\|y' - y\|} , y - p \right\rangle \left\langle \frac{y' - y}{\|y' - y\|}, x - p \right\rangle \\
& = \frac{1}{\|y' - y\|^2} \langle y' - y , y - p \rangle \langle y' - y , x - p \rangle
\end{split}
\end{equation*}
shows that $\langle y - p', x - p \rangle \leq 0$ in the case of $\langle y' - y , y - p \rangle \leq 0$ and $\langle y' - y , x - p \rangle \geq 0$. This can be shown using the Cauchy-Schwarz inequality, as
\begin{equation*}
\begin{split}
\left\langle y' - y , y - p \right\rangle & = \left\langle \frac{\|y\|}{\|x\|} x - y , y - \frac{\|x\|}{\|y\|} y \right\rangle
= \frac{\|y\|}{\|x\|} \langle x,y \rangle - \langle x,y \rangle - \| y \|^2 + \frac{\|x\|}{\|y\|} \| y \|^2
= \left(\frac{\|y\|}{\|x\|} - 1\right) \langle x,y \rangle - \| y \|^2 + \|x\| \|y\| \\
& \leq \left(\frac{\|y\|}{\|x\|} - 1\right) | \langle x,y \rangle | - \| y \|^2 + \|x\| \|y\|
\leq \left(\frac{\|y\|}{\|x\|} - 1\right) \|x\| \|y\| - \| y \|^2 + \|x\| \|y\| = 0
\end{split}
\end{equation*}
and
\begin{equation*}
\begin{split}
\langle y' - y , x - p \rangle & = \left\langle \frac{\|y\|}{\|x\|}x - y , x - \frac{\|x\|}{\|y\|} y \right\rangle
= \frac{\|y\|}{\|x\|} \|x\|^2 - \langle x,y \rangle - \langle y,x \rangle + \frac{\|x\|}{\|y\|} \|y\|^2 \\
& = \|y\| \|x\| - 2 \langle x,y \rangle + \|x\| \|y\| = 2 ( \|x\| \|y\| - \langle x,y \rangle ) \geq 0 .
\end{split}
\end{equation*}

We consider next the case $(0 < \|x\| < c \leq \|y\|)$. Let us assume first that $y$ and $x$ are linearly dependent, i.e. $y = \lambda x$ for some $\lambda \in \mathbb{R}$. Then it holds that
\begin{equation*}
| \underline{f}(y) - \underline{f}(x) | = \left| f\left(\frac{c \lambda x}{| \lambda | \| x \|}\right) - f(x) \right| \leq \Lambda \left\| x \left( \frac{\lambda c}{| \lambda | \| x \|} - 1 \right) \right\| \\
= \Lambda \left| \frac{\lambda c}{| \lambda | \| x \|}  - 1 \right| \| x \| =
\begin{cases}
    \Lambda(c - \| x \|) & , \ \lambda \geq 0 \\
    \Lambda(c + \| x \|) & , \ \lambda < 0
\end{cases}
\end{equation*}
Thus, if $\lambda \geq 0$, $\| x \| < \| y \|$ implies $\lambda > 1$ and we conclude that
\[ | \underline{f}(y) - \underline{f}(x) | \leq \Lambda(c - \| x \|) \leq \Lambda \| y - x \| , \]
as $\|y - x\| = (\lambda - 1) \|x\| = \|y\| - \|x\| \geq c - \| x \| $. Instead, if $\lambda < 0$, it holds that $\|y - x\| = |\lambda - 1| \|x\| = (|\lambda| + 1) \|x\| = \|y\| + \|x\| \geq c + \| x \| $, which implies in turn
\[ | \underline{f}(y) - \underline{f}(x) | \leq \Lambda(c + \| x \|) \leq \Lambda \| y - x \| . \]
Let us assume now that $x$ and $y$ are linearly independent. Then it holds that
\[ | \underline{f}(y) - \underline{f}(x) | = \left| f\left(\frac{c y}{\|y\|}\right) - f(x) \right| \leq \Lambda \left\| \frac{c y}{\|y\|} - x \right\| . \]
As $g : [0;1] \rightarrow \mathbb{R}_{\geq0}, t \mapsto \| (1-t)x + t y \|$ is continuous, the set $\mathcal{O} := \{ t \in [0;1] \mid g(t) < c \} = g^{-1}\big([0;c)\big)$ is open in $[0;1]$ and $\lambda := \sup \mathcal{O} \not \in \mathcal{O}$. Obviously, $\| g(t) \| \geq c$ for all $\lambda \leq t \leq 1$ and the continuity of $g$ yields $g(\lambda) = c$. (Moreover, the convexity of $\{ z \in \mathbb{R}^n \mid \| z \| < c \}$ asserts that $\| g(t) \| < c$ for all $0 \leq t < \lambda$.) Together with $x_0 := (1 - \lambda) x + \lambda y$, it holds then that
\begin{equation*}
\left\| \frac{c y}{\|y\|} - x \right\| \leq \left\| \frac{c y}{\|y\|} - x_0 \right\| + \left\| x_0 - x \right\| = \left\| \frac{\|x_0\|}{\|y\|} y - x_0 \right\| + \left\| x_0 - x \right\|
\end{equation*}
Note that the last subcase asserts that $\| \frac{\|x_0\|}{\|y\|} y - x_0 \| \leq \| y - x_0 \|$ if $y$ and $x_0$ are linearly independent. Let us assume towards a contradiction that $y$ and $x_0$ are linearly dependent. Then there exists a $\mu \neq 0$ such that $y = \mu x_0 = \mu (1 - \lambda) x + \mu \lambda y$, which implies in turn $(1 - \mu \lambda) y = \mu (1 - \lambda) x$. As $\lambda < 1$, this induces $x = \frac{1 - \mu \lambda}{\mu (1 - \lambda)} y$ and $x \neq 0$ yields the contradiction $y = \frac{\mu (1 - \lambda)}{1 - \mu \lambda} x$. Hence, it holds that $\| \frac{\|x_0\|}{\|y\|} y - x_0 \| \leq \| y - x_0 \|$ and it is sufficient to show that $\| y - x_0 \| + \| x_0 - x \| = \| y - x \|$. For this, we note that
\begin{equation*}
x_0 - x = - \lambda x + \lambda y = \lambda (y - x) \qquad \text{and} \qquad y - x_0 = y - (1 - \lambda) x - \lambda y = (1 - \lambda) (y - x)
\end{equation*}
imply together with
\begin{equation*}
\begin{split}
\|y - x \|^2 & = \langle y - x_0 + x_0 - x, y - x_0 + x_0 - x \rangle = \|y - x_0\|^2 + \|x_0 - x\|^2 + 2 \langle y - x_0, x_0 - x \rangle \\
& = \|y - x_0\|^2 + \|x_0 - x\|^2 + 2 \lambda (1 - \lambda) \langle y - x, y -  x \rangle = \|y - x_0\|^2 + \|x_0 - x\|^2 + 2 \| y - x_0 \| \| x_0 - x \| \\
& = (\|y - x_0\| + \|x_0 - x\|)^2
\end{split}
\end{equation*}
the desired relation.
\end{proof}


\begin{proof}[Proof of Theorem \ref{thm_scheme}]
The proof borrows ideas from convergence proofs of the Euler method \cite{Gear1971} and the heat equation \cite{thomas95}, respectively. Let us define
\begin{align*}
c & := \max_{(x,y,t) \in [0;1]^2 \times [0;T]} \ \ \max_{1 \leq l \leq L} \ \ \max_{1 \leq p \leq P} \ \ \max \Big\{ \ |\alpha_l(x,y,t)|, |\beta_p(x,y)| \ \Big\}, &
c' & := \sqrt{(L + P) (c + 1)^2},
\end{align*}
and $\underline{f_j} : \mathbb{R}^{L+P} \rightarrow \mathbb{R}$, where $1 \leq j \leq J$, by
\begin{equation*}
\begin{split}
& \underline{f_j}(\vec{z}) :=
\begin{cases}
f_j(\vec{z}) & \ , \| \vec{z} \| \leq c' \\
f_j\big(c' \frac{\vec{z}}{\| \vec{z} \|}\big) & \ , \text{otherwise.}
\end{cases}
\end{split}
\end{equation*}
Since $\| (z_1, \ldots, z_{L+P}) \|_\infty \leq c + 1$ implies that $\| (z_1, \ldots, z_{L+P}) \| \leq c'$, we infer that $\underline{f_j}(z_1, \ldots, z_{L + P}) = f_j(z_1, \ldots, z_{L + P})$ for all $\| (z_1, \ldots, z_{L+P}) \|_\infty \leq c + 1$. This shows that the solution of (\ref{eq_pde}) on $[0;1]^2 \times [0;T]$ is also a solution of the PDE system
\begin{equation}\label{eq_pde_proj}
\partial_t \alpha_l = \mu_l \bigtriangleup \alpha_l  + \sum_{1 \leq j \leq J} (d_{jl} - c_{jl}) \underline{f_j}(\alpha_1, \ldots, \alpha_L, \beta_1, \ldots, \beta_P),
\end{equation}
where $1 \leq l \leq L$, subjected to the DBCs.

Consequently, as it can be shown that (\ref{eq_pde_proj}) has a unique solution on $[0;1]^2 \times [0;T]$, the solution of (\ref{eq_pde}) on $[0;1]^2 \times [0;T]$ is the solution of (\ref{eq_pde_proj}) on $[0;1]^2 \times [0;T]$. To see this, let us assume towards a contradiction that $(\hat{\alpha}_1, \ldots, \hat{\alpha}_L)$ is a solution of (\ref{eq_pde_proj}) satisfying $\xi(t) > 0$ for some $t$, where $\xi : [0;T] \rightarrow \mathbb{R}_{\geq0}$ is given by $\xi(t) := \max_{(x,y) \in [0;1]^2} \max_{1 \leq l \leq L} | \alpha_l(x,y,t) - \hat{\alpha}_l(x,y,t) |$. As the function $\zeta : (x,y,t) \mapsto \max_{1 \leq l \leq L} | \alpha_l(x,y,t) - \hat{\alpha}_l(x,y,t) |$ is continuous on the compactum $[0;1]^2 \times [0;T]$, it is uniformly continuous on $[0;1]^2 \times [0;T]$ and thus it satisfies
\begin{equation}\label{eq_thm_scheme_unif_cont}
\forall \varepsilon > 0 \exists \delta > 0 \forall (x,y) \in [0;1]^2 \forall t, t' \in [0;T] \Big( |t - t'| < \delta \Rightarrow |\zeta(x,y,t) - \zeta(x,y,t')| < \varepsilon \Big).
\end{equation}
Since $\forall (x,y) \in [0;1]^2 \Big(|\zeta(x,y,t)| - \varepsilon < |\zeta(x,y,t')| < |\zeta(x,y,t)| + \varepsilon\Big)$ implies $\xi(t) - \varepsilon < \xi(t') < \xi(t) + \varepsilon$, (\ref{eq_thm_scheme_unif_cont}) yields the (uniform) continuity of $\xi$ on $[0;T]$. Consequently, the set $\mathcal{O} := \{ t \in [0;T] \mid \xi(t) > 0 \} = \xi^{-1}\big((0;\infty)\big)$ is open in $[0;T]$ and $t_0 := \inf \mathcal{O}$, which exists because $\mathcal{O} \neq \emptyset$, is not contained in $\mathcal{O}$. Also, since the solutions are continuous in $t$, it holds that $t_0 < T$. Since $\xi$ is continuous, there exists a $\delta > 0$ such that $| \xi(t') | = | \xi(t') - \xi(t_0) | < 1$ for all $t' \in [t_0; t_0 + \delta] \subseteq [0;T]$ and we get
\[ \forall 1 \leq l \leq L \forall t_0 \leq t' \leq t_0 + \delta \forall (x,y) \in [0;1]^2 \Big( | \hat{\alpha}_l(x,y,t') | \leq | \hat{\alpha}_l(x,y,t') - \alpha_l(x,y,t') | + | \alpha_l(x,y,t') | < 1 + c \Big) . \]
Consequently, $(\hat{\alpha}_1, \ldots, \hat{\alpha}_L)$ is a solution of (\ref{eq_pde}) on $[0;1]^2 \times [0; t_0 + \delta]$. The choice of $t_0$ ensures that $\xi(t') > 0$ for some $t_0 < t' \leq t_0 + \delta$. This is a contradiction to the uniqueness of the solution of (\ref{eq_pde}) on $[0;1]^2 \times [0; t_0 + \delta]$ and we infer that the solution of (\ref{eq_pde}) on $[0;1]^2 \times [0;T]$ can be recovered by solving (\ref{eq_pde_proj}).

We show next that the finite difference scheme of (\ref{eq_pde_proj})
\begin{align}\label{eq_thm_scheme_ode_euler_proj}
\underline{\mathcal{L}}^{(x,y)}_l(\underline{\vec{a}}) & := \sum_{1 \leq j \leq J} (d_{jl} - c_{jl}) \underline{f_j}\big(\underline{a}^{(x,y)}_1, \ldots, \underline{a}^{(x,y)}_L, \beta_1(x,y), \ldots, \beta_P(x,y)\big) , \quad \underline{\vec{a}} \in \mathbb{R}^{\mathcal{R} \times \{1, \ldots, L\}} \nonumber \\
\underline{a}_l^{(x,y)}(m+1) & := \underline{a}_l^{(x,y)}(m) + \mathds{1}(x,y) \cdot \dt \cdot \Big( \mu_l \lapl \underline{a}_l^{(x,y)}(m) + \underline{\mathcal{L}}^{(x,y)}_l\big(\underline{\vec{a}}(m)\big) \Big), \qquad m \geq 0 ,
\end{align}
is convergent in the supremum norm $\| \underline{a}(m) \|_\infty := \max_{(x,y) \in \mathcal{R}, 1 \leq l \leq L} | \underline{a}_l^{(x,y)}(m)|$ on $\mathbb{R}^{\mathcal{R} \times \{1, \ldots, L\}}$. For this, we observe first that each $\underline{f_j}$ is Lipschitz with respect to $\| \cdot \|$ by Proposition \ref{prop_loc_glob_lip}. This and the equivalence of all norms on a finite dimensional vector space show that each $\underline{f_j}$ is also Lipschitz with respect to $\| \cdot \|_\infty$. Let us denote the space step error and the time step error of the scheme at $(x,y,t) \in \mathcal{R} \times [0; T - \dt]$ by
\begin{equation*}
\begin{split}
\rho_l(x,y,t) & := \mu_l \bigtriangleup \alpha_l(x,y,t) - \mu_l \lapl \alpha_l(x,y,t) \\
\eta_l(x,y,t) & := \alpha_l(x,y,t + \dt) - \alpha_l(x,y,t) - \dt\big( \mu_l \bigtriangleup \alpha_l(x,y,t) + g_l(x,y,t) \big) ,
\end{split}
\end{equation*}
respectively, where $g_l : [0;1]^2 \times [0;T] \rightarrow \mathbb{R}$ is defined as $\dsl g_l := \sum_{1 \leq j \leq J} (d_{jl} - c_{jl}) \underline{f_j}(\alpha_1, \ldots, \alpha_L, \beta_1, \ldots, \beta_P) .$

We begin with bounding the space step error. For this, we fix a region $(x,y) \in \mathcal{R} \setminus \Omega$ and note that there exist $x' \in (x-\ds;x+\ds)$ and $y' \in (y-\ds;y+\ds)$ such that
\begin{equation*}
\begin{split}
\frac{\rho_l(x,y,t)}{\mu_l} & = \partial_{xx} \alpha_l (x,y,t) + \partial_{yy} \alpha_l (x,y,t) \\
& \qquad - \bigg( \frac{\alpha_l(x+\ds,y, t) - 2 \alpha_l(x,y,t) + \alpha_l(x-\ds,y,t)}{\ds^2} + \frac{\alpha_l(x,y + \ds, t) - 2 \alpha_l(x,y,t) + \alpha_l(x,y-\ds,t)}{\ds^2} \bigg) \\
& = \frac{\partial_{xxxx} \alpha_l(x',y,t)}{12} \ds^2 + \frac{\partial_{yyyy} \alpha_l(x,y',t)}{12} \ds^2.
\end{split}
\end{equation*}
Consequently, there exists a $C_1 \in \mathbb{R}_{\geq0}$ which \emph{does not depend} on $\dt, \ds$ and satisfies $|\rho_l(x,y,t)| \leq C_1 \ds^2$ for all $(x,y) \in \mathcal{R} \setminus \Omega$, $0 \leq t \leq T$ and $1 \leq l \leq L$.

We turn now to the time step error. Applying the mean value theorem to $t \mapsto \alpha_l(x,y,t)$ allows us to infer the existence of a $t \leq t' \leq t + \dt$ such that $\alpha_l(x,y,t + \dt) - \alpha_l(x,y,t) = \partial_t \alpha_l(x,y,t') \dt$. Hence
\begin{equation*}
\begin{split}
\eta_l(x,y,t) & = \alpha_l(x,y,t + \dt) - \alpha_l(x,y,t) - \dt\big( \mu_l \bigtriangleup \alpha_l(x,y,t) + g_l(x,y,t) \big)
= \dt \big( \partial_t \alpha_l(x,y,t') - \mu_l \bigtriangleup \alpha_l(x,y,t) - g_l(x,y,t) \big) \\
& = \dt \Big( \mu_l \bigtriangleup \alpha_l(x,y,t') + g_l(x,y,t') - \mu_l \bigtriangleup \alpha_l(x,y,t) - g_l(x,y,t) \Big) .
\end{split}
\end{equation*}
Let $\Lambda' > 0$ refer to the Lipschitz constant of $\mu_1 \bigtriangleup \alpha_1, \ldots, \mu_L \bigtriangleup \alpha_L$ on $[0;1]^2 \times [0;T]$ and $\Lambda'' > 0$ be the \emph{global} Lipschitz constant of $\underline{f_1}, \ldots, \underline{f_J}$. Then, together with $\dsl \Lambda := \max_{1 \leq l \leq L} \sum_{1 \leq j \leq J} |d_{jl} - c_{jl}| \Lambda''$, the above calculation yields
\begin{multline*}
\frac{| \eta_l(x,y,t) |}{\dt} \leq \sum_{1 \leq j \leq J} |d_{jl} - c_{jl}| \Lambda'' \| \vec{\alpha}(x,y,t') - \vec{\alpha}(x,y,t) \|_\infty + \Lambda' \| (x,y,t') - (x,y,t) \|_\infty \\
\leq \Lambda \| \vec{\alpha}(x,y,t') - \vec{\alpha}(x,y,t) \|_\infty + \Lambda' \dt
\leq \Lambda \dt \max_{1 \leq l' \leq L} \max_{(x,y,t) \in [0;1]^2 \times [0;T]} | \partial_t \alpha_{l'}(x,y,t) | + \Lambda' \dt ,
\end{multline*}
where the last inequality follows again by means of the mean value theorem. Hence, there exists a $C_2 \in \mathbb{R}_{\geq0}$ which \emph{does not depend} on $\dt, \ds$ and satisfies $| \eta_l(x,y,t) | \leq C_2 \dt^2$ for all $(x,y) \in \mathcal{R}$, $0 \leq t \leq T - \dt$ and $1 \leq l \leq L$.

After we were able to derive bounds for $|\eta_l|$ and $|\rho_l|$, we note that
\begin{equation*}
\eta_l(x,y,t) = \alpha_l(x,y,t + \dt) - \alpha_l(x,y,t) - \dt\big( \mu_l \lapl \alpha_l(x,y,t) + \rho_l(x,y,t) + g_l(x,y,t) \big)
\end{equation*}
is equivalent to
\begin{equation*}
\alpha_l(x,y,t + \dt) = \alpha_l(x,y,t) + \dt \Big( \mu_l \lapl \alpha_l(x,y,t) + g_l(x,y,t)\Big) + \dt \rho_l(x,y,t) + \eta_l(x,y,t) .
\end{equation*}
Consequently, together with $t_m := m \dt$, $e_l^{(x,y)}(m) := \alpha_l(x,y,t_m) - \underline{a}_l^{(x,y)}(m)$ and
\begin{equation*}
h_l(x,y,m) := \sum_{1 \leq j \leq J} (d_{jl} - c_{jl}) \underline{f_j}\big(\underline{a}^{(x,y)}_1(m), \ldots, \underline{a}^{(x,y)}_L(m), \beta_1(x,y), \ldots, \beta_P(x,y)\big),
\end{equation*}
it holds that
\begin{multline*}
|e_l^{(x,y)}(m+1)| \leq \Big|(1 - |\mathcal{N}(x,y)| r_l) e_l^{(x,y)}(m) + r_l \sum_{(\tilde{x},\tilde{y}) \in \mathcal{N}(x,y)} e_l^{(\tilde{x},\tilde{y})}(m)
+ \dt \Big(g_l(x,y,t_m) - h_l(x,y,m) \Big)\Big| + \\
+ \dt |\rho_l(x,y,t_m)| + |\eta_l(x,y,t_m)| \leq \| \vec{e}(m) \|_\infty + \dt \Lambda \| \vec{e}(m) \|_\infty + \dt C_1 \ds^2 + C_2 \dt^2
\end{multline*}
for all $0 \leq m \leq M - 1$ and $(x,y) \in \mathcal{R} \setminus \Omega$, thanks to $r_l \leq 1/4$. Noting that $e_l^{(x,y)}(m + 1) = 0$ for all $(x,y) \in \Omega$ because of the DBCs, we conclude that the estimation
\[ |e_l^{(x,y)}(m+1)| \leq (1 + \dt \Lambda) \| \vec{e}(m) \|_\infty + \dt C_1 \ds^2 + C_2 \dt^2 \]
holds for all $(x,y) \in \mathcal{R}$. Together with $C_3 := \max\{C_1,C_2\}$, this yields
\begin{equation*}
\begin{split}
\!\! \|\vec{e}(m)\|_\infty & \leq (1 + \dt \Lambda)^m \| \vec{e}(0) \|_\infty + C_3(\dt\ds^2 + \dt^2) \sum_{k = 0}^{m-1} (1 + \dt \Lambda)^k
\leq (1 + \dt \Lambda)^m \| \vec{e}(0) \|_\infty +  C_3(\dt\ds^2 + \dt^2) \frac{(1 + \dt \Lambda)^m - 1}{\dt \Lambda} \\
& \leq e^{m \dt \Lambda} \| \vec{e}(0) \|_\infty +  C_3(\ds^2 + \dt) \frac{e^{m \dt \Lambda} - 1}{\Lambda}
\leq e^{T \Lambda} \| \vec{e}(0) \|_\infty +  C_3(\ds^2 + \dt) \frac{e^{T \Lambda} - 1}{\Lambda}
= 0 + \frac{e^{T \Lambda} - 1}{\Lambda/C_3} (\ds^2 + \dt)
\end{split}
\end{equation*}
for all $0 \leq m \leq M$. This shows that (\ref{eq_thm_scheme_ode_euler_proj}) converges in the supremum norm $\| \cdot \|_\infty$ to the solution of (\ref{eq_pde_proj}) and (\ref{eq_pde}) on $[0;1]^2 \times [0;T]$. Moreover, it holds that
\begin{equation*}
| \underline{a}^{(x,y)}_l(m) | \leq | \underline{a}^{(x,y)}_l(m) - \alpha_l(x,y,t_m) | + | \alpha_l(x,y,t_m) | \leq \| \vec{e}(m) \|_\infty + c \leq 1/2 + c
\end{equation*}
for all $0 \leq m \leq M$, $1 \leq l \leq L$ and $(x,y) \in \mathcal{R}$ if $\ds$ and $\dt$ are sufficiently small. Specifically, there exist $K_0$ and $M_0$ such that $\underline{\mathcal{L}}^{(x,y)}_l(\underline{\vec{a}}(m)) = \mathcal{L}^{(x,y)}_l(\underline{\vec{a}}(m))$ for all $0 \leq m \leq M$, $(x,y) \in \mathcal{R}$ and $1 \leq l \leq L$, if $K \geq K_0$ and $M \geq M_0$. Consequently, (\ref{eq_thm_scheme_ode_euler_proj}) coincides with the scheme (\ref{eq_ode_euler}) in the case of $K \geq K_0$ and $M \geq M_0$ and the proof is complete.
\end{proof}


\begin{proof}[Proof of Theorem \ref{thm_main}]
In the following, we will use the auxiliary results from the proof of Theorem \ref{thm_scheme}. Recall that the sequence induced by the finite difference scheme (\ref{eq_thm_scheme_ode_euler_proj}) remains bounded by $c + 1/2$ in the supremum norm if $M \geq M_0$ and $K \geq K_0$ and that (\ref{eq_thm_scheme_ode_euler_proj}) can be interpreted as the Euler method of the ODE system
\begin{equation}\label{eq_thm_main_ode_proj}
\frac{d}{dt} \underline{a}^{(x,y)}_l(t) = \mathds{1}(x,y) \cdot \Big(\mu_l \lapl \underline{a}_l^{(x,y)} + \underline{\mathcal{L}}^{(x,y)}_l(\underline{\vec{a}}(t)) \Big) ,
\end{equation}
where $(x,y) \in \mathcal{R}$ and $1 \leq l \leq L$. The global Lipschitz continuity of (\ref{eq_thm_main_ode_proj}) and the global version of Picard-Lindel\"{o}f's theorem assert that (\ref{eq_thm_main_ode_proj}) has a unique solution $\underline{\vec{a}}$ which is defined on the whole real line. Define $C'_1 := \frac{e^{T \Lambda} - 1}{\Lambda/C_3}$ and fix a $K$ such that $K \geq K_0$ and $C'_1 \ds^2 < \varepsilon / 12$. The uniform continuity of $\underline{\vec{a}}$ on $[0;T]$ implies the existence of a $\delta_1 > 0$ such that
\[ \forall 1 \leq l \leq L \forall (x,y) \in \mathcal{R} \forall t,t' \in [0;T] \Big( |t - t'| < \delta_1 \Rightarrow | \underline{a}_l^{(x,y)}(t) - \underline{a}_l^{(x,y)}(t') | < \varepsilon / 12 \Big). \]
Similarly, thanks to the uniform continuity of $(\alpha_1, \ldots, \alpha_L)$ on $[0;1]^2 \times [0;T]$, there exists a $\delta_2 > 0$ such that
\[ \forall 1 \leq l \leq L \forall (x,y) \in [0;1]^2 \forall t,t' \in [0;T] \Big( |t - t'| < \delta_2 \Rightarrow | \alpha_l(x,y,t) - \alpha_l(x,y,t') | < \varepsilon / 12 \Big) . \]
Moreover, by~\cite{Gear1971} there exists a $C'_2 > 0$ such that
\[ \max_{0 \leq m \leq M} \max_{(x,y) \in \mathcal{R}, 1 \leq l \leq L} | \underline{a}_l^{(x,y)}(t_m) - \underline{a}_l^{(x,y)}(m) | \leq C'_2 \dt \]
for all $M \geq M_0$. Note that $\underline{a}_l^{(x,y)}(t_m)$ denotes the exact ODE solution at time $t_m$, whereas  $\underline{a}_l^{(x,y)}(m)$ denotes the $m$-th sequence element in the method of Euler (\ref{eq_thm_scheme_ode_euler_proj}). In the following, let us fix an $M$ such that $\dsl \dt = \frac{T}{M} < \min\left\{\frac{T}{M_0}, \delta_1, \delta_2, \frac{\varepsilon}{4(C'_1 + C'_2)}\right\}$ and assume without loss of generality that $\varepsilon < 1$. Then, for an arbitrary $t \in [0;T]$ and an $0 \leq m \leq M$ such that $|t - t_m| < \dt$, it holds that
\begin{equation*}
| \underline{a}_l^{(x,y)}(t) | \leq | \underline{a}_l^{(x,y)}(t) - \underline{a}_l^{(x,y)}(t_m) | + | \underline{a}_l^{(x,y)}(t_m) - \underline{a}_l^{(x,y)}(m) |  + | \underline{a}_l^{(x,y)}(m) |
\leq \varepsilon / 12 + C'_2 \dt + (c + 1/2) \leq \varepsilon / 12 + \varepsilon / 4 + c + 1/2
\end{equation*}
for all $(x,y) \in \mathcal{R}$ and $1 \leq l \leq L$, meaning that the solution of (\ref{eq_thm_main_ode_proj}) on $[0;T]$ is bounded by $c + 1$ in the supremum norm. Consequently, since $\| (z_1, \ldots, z_{L+P}) \|_\infty \leq c + 1$ implies $\underline{f_j}(z_1, \ldots, z_{L + P}) = f_j(z_1, \ldots, z_{L + P})$ for all $1 \leq j \leq J$, it holds that the solution $\vec{a}$ of the original ODE system (\ref{eq_ode_sys}) is well-defined on $[0;T]$ and satisfies $\vec{a}_{|[0;T]} = \underline{\vec{a}}_{|[0;T]}$. Hence, for an arbitrary $t \in [0;T]$ and an $0 \leq m \leq M$ such that $|t - t_m| < \dt$, it holds that
\begin{equation*}
\begin{split}
| a_l^{(x,y)}(t) - \alpha_l(x,y,t) | & \leq | \underline{a}_l^{(x,y)}(t) - \underline{a}_l^{(x,y)}(t_m) | + | \underline{a}_l^{(x,y)}(t_m) - \underline{a}_l^{(x,y)}(m) |
+ | \underline{a}_l^{(x,y)}(m) - \alpha_l(x,y,t_m) | + | \alpha_l(x,y,t_m) - \alpha_l(x,y,t) | \\
& \leq \varepsilon/12 + C'_2\dt + C'_1(\ds^2 + \dt) + \varepsilon/12 \leq \varepsilon/4 + (C'_1 + C'_2) \dt \leq \varepsilon / 2
\end{split}
\end{equation*}
for all $(x,y) \in \mathcal{R}$ and $1 \leq l \leq L$. Together with
\begin{equation*}
\begin{split}
\phi_1 & := \sup_{0 \leq t \leq T} \max_{(x,y) \in \mathcal{R}, 1 \leq l \leq L} \left| \frac{1}{N} A_l^{(x,y)}(t) - \alpha_l(x,y,t) \right| ,  \\
\phi_2 & := \sup_{0 \leq t \leq T} \max_{(x,y) \in \mathcal{R}, 1 \leq l \leq L} \left| \frac{1}{N} A_l^{(x,y)}(t) - a_l^{(x,y)}(t) \right| , \\
\phi_3 & := \sup_{0 \leq t \leq T} \max_{(x,y) \in \mathcal{R}, 1 \leq l \leq L} \left| a_l^{(x,y)}(t) - \alpha_l(x,y,t) \right| ,
\end{split}
\end{equation*}
we infer that $\mathbb{P} \{ \phi_1 > \varepsilon \} \leq \mathbb{P} \{ \phi_2 + \phi_3 > \varepsilon \} \leq \mathbb{P} \{ \phi_2 > \varepsilon/2 \ \lor \ \phi_3 > \varepsilon/2 \}
\leq \mathbb{P}\{ \phi_2 > \varepsilon/2 \} + \mathbb{P}\{ \phi_3 > \varepsilon/2 \}$. Since the choice of $K$ from above implies that $\phi_3 \leq \varepsilon/2$, it holds that $\mathbb{P}\{ \phi_3 > \varepsilon/2 \} = 0$ and Theorem \ref{thm_kurtz} yields the claim.
\end{proof}


\begin{proof}[Proof of Theorem~\ref{thm_scheme_nbcs}]
In the following, we will use results and variables from the proof of Theorem~\ref{thm_scheme}. We first discuss the space step error and start by fixing an $(x,y) \in \Omega$ where $x = 0$ and $\ds \leq y \leq 1 - \ds$. Due to Taylor's theorem there exist $x' \in (0;\ds)$ and $y' \in (y -\ds;y + \ds)$ such that
\begin{equation*}
\begin{split}
\frac{\rho_l(0,y,t)}{\mu_l} & = \partial_{xx} \alpha_l (0,y,t) + \partial_{yy} \alpha_l (0,y,t) - \Big( \frac{\alpha_l(0 + \ds,y,t) - \alpha_l(0,y,t)}{\ds^2} + \frac{\alpha_l(0,y + \ds, t) - 2 \alpha_l(0,y,t) + \alpha_l(0,y-\ds,t)}{\ds^2} \Big) \\
& = \partial_{xx} \alpha_l (0,y,t) - \Big( \frac{\partial_x \alpha_l(0,y,t)}{\ds} + \frac{\partial_{xx} \alpha_l(0,y,t)}{2!} + \frac{\partial_{xxx} \alpha_l(x',y,t)}{3!} \ds - \frac{\partial_{yyyy} \alpha_l(0,y',t)}{12} \ds^2 \Big) \\
& = \frac{\partial_{xx} \alpha_l(0,y,t)}{2} + \ds \Big( \frac{\partial_{xxx} \alpha_l(x',y,t)}{3!} - \frac{\partial_{yyyy} \alpha_l(\ds,y',t)}{12} \ds \Big)
\end{split}
\end{equation*}
As a similar calculation applies in the other seven boundary cases, we infer that there exists a $C'_1 \in \mathbb{R}_{\geq0}$ which \emph{does not depend} on $\dt, \ds$ and satisfies $|\rho_l(x,y,t)| \leq C'_1$ for all $(x,y) \in \Omega$, $0 \leq t \leq T$ and $1 \leq l \leq L$. By choosing $C_1 \geq C_1'$, we may assume without loss of generality that $C_1 \geq C'_1$.

In order to show consistency with respect to $\| \cdot \|_{1,\ds}$, we have to prove that
\begin{equation}\label{eq_thm_scheme_nbcs_1}
\Big\| \frac{\vec{\eta}(x,y,t)}{\dt} + \vec{\rho}(x,y,t) \Big\|_{1,\ds} \rightarrow 0, \quad \dt, \ds \rightarrow 0
\end{equation}
For this, we define
\[ C_{(x,y)} :=
\begin{cases}
C'_1       & , \ (x,y) \in \Omega \\
C_1 \ds^2 & , \ (x,y) \in \mathcal{R} \setminus \Omega
\end{cases}
\]
conclude (by using results from the proof of Theorem~\ref{thm_scheme})
\begin{equation*}
\Big\| \frac{\vec{\eta}(x,y,t)}{\dt} + \vec{\rho}(x,y,t) \Big\|_{1,\ds} \leq \sum_{l = 1}^L \sum_{(x,y) \in \mathcal{R}} (\dt C_2 + C_{(x,y)}) \ds^2
\leq L \Big( 4 C_2 \dt + \sum_{(x,y) \in \mathcal{R} \setminus \Omega} (C_2 \ds^2) \ds^2 + \sum_{(x,y) \in \Omega} C_1 \ds^2 \Big)
\end{equation*}
Hence, as $|\Omega| = 2(K + 1) + 2(K - 1) = 4K$ and $| \mathcal{R} \setminus \Omega | = (K - 1)^2$ we infer that $|\Omega| \ds^2 = |\Omega| / K^2 = 4/K = 4 \ds$ and $| \mathcal{R} \setminus \Omega | \ds^2 \leq 1$. Consequently, it holds
\begin{equation*}
L \Big( 4 C_2 \dt + \sum_{(x,y) \in \mathcal{R} \setminus \Omega} (C_2 \ds^2) \ds^2 + \sum_{(x,y) \in \Omega} C_1 \ds^2 \Big) \leq L \Big( 4 C_2 \dt + C_2 \ds^2 + 4C_1 \ds \Big) ,
\end{equation*}
yielding the consistency (\ref{eq_thm_scheme_nbcs_1}).

Next, we show that the scheme is stable with respect to $\| \cdot \|_\infty$. (Note that this is sufficient, as $\| \cdot \|_{1,\ds} \leq 4 L \| \cdot \|_\infty$.) From the proof of Theorem~\ref{thm_scheme} we know that
\[ \dt |\rho_l(x,y,t_m)| + |\eta_l(x,y,t_m)| \leq C_1 \dt + C_2 \dt^2 \]
if $\ds$ is sufficiently small. By performing a similar calculation as in the case of DBCs, we infer for $C'_3 := \max\{C_1, C_2\}$
\begin{equation*}
|e_l^{(x,y)}(m)| \leq \| \vec{e}(m) \|_\infty + \dt \Lambda \| \vec{e}(m) \|_\infty + C'_3(\dt + \dt^2) \leq (1 + \dt \Lambda) \| \vec{e}(m) \|_\infty + C'_3(\dt + \dt^2),
\end{equation*}
which allows us to conclude, for all $0 \leq m \leq M$, that
\begin{align*}
\|\vec{e}(m)\|_\infty & \leq (1 + \dt \Lambda)^m \| \vec{e}(0) \|_\infty + C'_3(\dt + \dt^2) \sum_{k = 0}^{m-1} (1 + \dt \Lambda)^k
\leq (1 + \dt \Lambda)^m \| \vec{e}(0) \|_\infty +  C'_3(\dt + \dt^2) \frac{(1 + \dt \Lambda)^m - 1}{\dt \Lambda} \\
& \leq e^{m \dt \Lambda} \| \vec{e}(0) \|_\infty +  C'_3(1 + \dt) \frac{e^{m \dt \Lambda} - 1}{\Lambda}
\leq e^{T \Lambda} \| \vec{e}(0) \|_\infty +  C'_3(1 + \dt) \frac{e^{T \Lambda} - 1}{\Lambda} \\
& = 0 + \frac{e^{T \Lambda} - 1}{\Lambda/C'_3} (1 + \dt)
\end{align*}
Consequently, as $\vec{\alpha}$ is bounded on $[0;1]^2 \times [0;T]$,
\begin{equation*}
\begin{split}
| a_l^{(x,y)}(m) -  a_l^{(x,y)}(0) | & \leq | a_l^{(x,y)}(m) - a_l(x,y,t_m) | + | \alpha_l(x,y,t_m)  - a_l^{(x,y)}(0) |
\end{split}
\end{equation*}
and $\| \cdot \|_{1,\ds} \leq 4 L \| \cdot \|_\infty$, the scheme is stable with respect to $\| \cdot \|_{1,\ds}$.
\end{proof}


\begin{proof}[Proof of Corollary \ref{cor_scheme_nbcs}]
From the proof of Theorem~\ref{thm_scheme_nbcs} it becomes apparent that the scheme (\ref{eq_ode_euler_nbcs}) is accurate of order $(1,1)$~\cite[Definition 2.3.3]{thomas95} with respect to $\lVert \cdot \rVert_{\ds,1}$. Since Theorem~\ref{thm_scheme_nbcs} implies also that (\ref{eq_ode_euler_nbcs}) is stable with respect to $\lVert \cdot \rVert_{\ds,1}$, the proof of Lax's theorem~\cite[Theorem 2.5.2]{thomas95} yields the claim.
\end{proof}


\begin{proof}[Proof of Theorem \ref{thm_main_nbcs}]
Theorem~\ref{thm_scheme_nbcs} ensures that there exist $K_0, M_0 \geq 0$ and $C > 0$ such that
\[ \max_{0 \leq m \leq M} \frac{1}{K^2} \sum_{\substack{(x,y) \in \mathcal{R}_K \\ 1 \leq l \leq L}} | \alpha_l(x,y,m\dt) - a^{(x,y)}_l(m) | \leq C(\ds + \dt)  \]
for all $K \geq K_0$ and $M \geq M_0$. Moreover, the ODE system
\begin{equation}\label{eq_thm_main_nbcs_1}
\frac{d}{dt} a^{(x,y)}_l(t) = \mu_l \lapl a_l^{(x,y)} + \mathcal{L}^{(x,y)}_l(\vec{a}(t)),
\end{equation}
where $(x,y) \in \mathcal{R}$ and $1 \leq l \leq L$, is globally Lipschitz. Thus, the global version of Picard-Lindel\"{o}f's theorem assert that (\ref{eq_thm_main_nbcs_1}) has a unique solution $\vec{a}$ which is defined on the whole real line. The uniform continuity of $\vec{a}$ on $[0;T]$ implies the existence of a $\delta_1 > 0$ such that
\[ \forall 1 \leq l \leq L \forall (x,y) \in \mathcal{R} \forall t,t' \in [0;T] \Big( |t - t'| < \delta_1 \Rightarrow | a_l^{(x,y)}(t) - a_l^{(x,y)}(t') | < \varepsilon / 48 L \Big). \]
Similarly, thanks to the uniform continuity of $(\alpha_1, \ldots, \alpha_L)$ on $[0;1]^2 \times [0;T]$, there exists a $\delta_2 > 0$ such that
\[ \forall 1 \leq l \leq L \forall (x,y) \in [0;1]^2 \forall t,t' \in [0;T] \Big( |t - t'| < \delta_2 \Rightarrow | \alpha_l(x,y,t) - \alpha_l(x,y,t') | < \varepsilon / 48 L \Big) . \]
Moreover, by~\cite{Gear1971} there exists a $C' > 0$ such that
\[ \max_{0 \leq m \leq M} \max_{(x,y) \in \mathcal{R}, 1 \leq l \leq L} | a_l^{(x,y)}(t_m) - a_l^{(x,y)}(m) | \leq C' \dt \]
for all $M \geq M_0$. Note that $a_l^{(x,y)}(t_m)$ denotes the exact ODE solution at time $t_m$, whereas  $a_l^{(x,y)}(m)$ denotes the $m$-th sequence element in the method of Euler (\ref{eq_thm_main_nbcs_1}).
In the following, let us fix an $M \geq M_0$ such that $\dsl \dt = T / M < \min\left\{\delta_1, \delta_2, \varepsilon / 48 L C, \varepsilon / 48 L C' \right\}$ and $K \geq K_0$ such that $C \ds \leq \varepsilon / 48 L$. Then, for an arbitrary $t \in [0;T]$ and an $0 \leq m \leq M$ such that $|t - t_m| < \dt$, it holds that
\begin{equation*}
\begin{split}
| a_l^{(x,y)}(t) - \alpha_l(x,y,t) | & \leq | {a}_l^{(x,y)}(t) - {a}_l^{(x,y)}(t_m) | + | {a}_l^{(x,y)}(t_m) - {a}_l^{(x,y)}(m) |
+ | {a}_l^{(x,y)}(m) - \alpha_l(x,y,t_m) | + | \alpha_l(x,y,t_m) - \alpha_l(x,y,t) | \\
& \leq \varepsilon / 48 L + \varepsilon / 48 L  + \varepsilon / 24 L  + \varepsilon / 48 L \leq \varepsilon / 8 L
\end{split}
\end{equation*}
for all $(x,y) \in \mathcal{R}$ and $1 \leq l \leq L$. Together with
\begin{equation*}
\begin{split}
\phi_1 & := \sup_{0 \leq t \leq T} \frac{1}{K^2} \sum_{\substack{(x,y) \in \mathcal{R}_K \\ 1 \leq l \leq L}} \left| \frac{1}{N} A_l^{(x,y)}(t) - \alpha_l(x,y,t) \right| ,  \\
\phi_2 & := \sup_{0 \leq t \leq T} \frac{1}{K^2} \sum_{\substack{(x,y) \in \mathcal{R}_K \\ 1 \leq l \leq L}} \left| \frac{1}{N} A_l^{(x,y)}(t) - a_l^{(x,y)}(t) \right| , \\
\phi_3 & := \sup_{0 \leq t \leq T} \frac{1}{K^2} \sum_{\substack{(x,y) \in \mathcal{R}_K \\ 1 \leq l \leq L}} \left| a_l^{(x,y)}(t) - \alpha_l(x,y,t) \right| ,
\end{split}
\end{equation*}
we infer that $\mathbb{P} \{ \phi_1 > \varepsilon \} \leq \mathbb{P} \{ \phi_2 + \phi_3 > \varepsilon \} \leq \mathbb{P} \{ \phi_2 > \varepsilon/2 \ \lor \ \phi_3 > \varepsilon/2 \}
\leq \mathbb{P}\{ \phi_2 > \varepsilon/2 \} + \mathbb{P}\{ \phi_3 > \varepsilon/2 \}$. Since the choice of $M, K$ from above implies that $\phi_3 \leq (\varepsilon/8 L) 3 L \leq \varepsilon / 2$, it holds that $\mathbb{P}\{ \phi_3 > \varepsilon/2 \} = 0$. Consequently, Theorem \ref{thm_kurtz_nbcs} and $\lVert \cdot \rVert_{\ds, 1} \leq 4 L \lVert \cdot \rVert_\infty$ yield the claim.
\end{proof}

\end{document}